\documentclass[twocolumn,twoside]{IEEEtran}
\usepackage{times} 
\usepackage{epsfig}
\usepackage{graphicx}
\usepackage{epstopdf}
\usepackage{color,subfigure, latexsym, amsmath, amsfonts, amssymb,cite}
\usepackage{algorithm}
\usepackage{algorithmic}
\usepackage{epsfig}
\usepackage{graphicx}
\usepackage{epstopdf}
\usepackage{bm,amsmath,amssymb,amsfonts,graphicx,epsfig,amsthm,color, xcolor}
\usepackage{bbold,dsfont}
\usepackage{float}

\usepackage[hyphens]{url}
\PassOptionsToPackage{bookmarks=false}{hyperref}

\usepackage[depth=-1]{bookmark}

\usepackage{booktabs}       
\usepackage{amsfonts}       
\usepackage{microtype}      

\usepackage{times}
\usepackage{graphicx} 

\usepackage{wrapfig}



\usepackage{accents}
\newtheorem{lemma}{Lemma}

\newtheorem{theorem}{Theorem}
\newtheorem{definition}{Definition}
\newtheorem{assumption}{Assumption}
\theoremstyle{remark}\newtheorem{remark}{Remark}

\allowdisplaybreaks

\title{\LARGE \bf
	Data-driven Self-triggered 
	Control
	 via Trajectory Prediction
}
\author{
		Wenjie~Liu,~Jian~Sun,~\IEEEmembership{Senior Member,~IEEE}, Gang Wang,~\IEEEmembership{Member,~IEEE},
		\\Francesco Bullo,~\IEEEmembership{Fellow,~IEEE},~and
		Jie Chen,~\IEEEmembership{Fellow,~IEEE}
		\thanks{The work was supported in part by the National Key R\&D Program of China under Grant 2021YFB1714800, in part by the National Natural Science Foundation of China under Grants 62088101, 62173034, 61925303, 
			 and in part by the Chongqing Natural Science Foundation under Grant 2021ZX4100027. (\emph{Corresponding author}: Gang Wang.)
		}
		\thanks{
			W. Liu, J. Sun, and G. Wang are with the State Key Lab of Intelligent Control and Decision of Complex Systems and the School of Automation, Beijing Institute of Technology, Beijing 100081, China, and also with the Beijing Institute of Technology Chongqing Innovation Center, Chongqing 401120, China (e-mail: liuwenjie@bit.edu.cn; sunjian@bit.edu.cn;  gangwang@bit.edu.cn).
			
			F. Bullo is with the Mechanical Engineering Department and the Center of Control, Dynamical Systems and Computation, UC Santa Barbara, CA 93106-5070, USA (e-mail: bullo@ucsb.edu).
			
			J. Chen is with the Department of Control Science and Engineering, Tongji University, Shanghai 201804, China, and also with the State Key Lab of Intelligent Control and Decision of Complex Systems, Beijing Institute of Technology, Beijing 100081, China 	
			(e-mail: chenjie@bit.edu.cn).
		}
}
\begin{document}
	\maketitle

	\begin{abstract}
		Self-triggered control, a well-documented technique for reducing the communication overhead while ensuring desired system performance, is gaining increasing popularity. 
		However, existing methods for self-triggered control require explicit system models that are assumed perfectly known \emph{a priori}.
		An end-to-end control paradigm known as data-driven control learns control laws directly from data, and offers a competing alternative to the routine system identification-then-control method. 
		In this context, the present paper puts forth data-driven self-triggered control schemes for unknown linear systems using data collected offline. Specifically, for output feedback control systems, a data-driven model predictive control (MPC) scheme is proposed, which computes a sequence of control inputs while generating a predicted system trajectory.		
		A data-driven self-triggering law is designed using the predicted trajectory, to determine the next triggering time once a new measurement becomes available.  
		For state feedback control systems, instead of capitalizing on MPC
		to predict the trajectory, 
		a data-fitting problem using the pre-collected input-state data is solved, whose solution is employed to construct the self-triggering mechanism. 
		Both feasibility and stability are established for the proposed self-triggered controllers, which are validated using numerical examples.
	\end{abstract}
	
	\section{Introduction}\label{sec:intro}
Thanks to recent advances in data acquisition and computing technologies,  data-driven control has attracted considerable attention in the past years.  
	Designing control laws directly from data without resorting to any system identification step, offers an appealing alternative to the traditional model-based control paradigm \cite{Hou2013,persis2020data}. This is because in real-world applications, it is always difficult or even impossible to acquire an accurate system model \cite{Astrom1973article,eng2022,LiuG2022}. 
	Indeed, a number of publications are devoted to studying data-driven control. 
	These were mainly inspired by the celebrated Fundamental Lemma proposed in \cite{willems2005note},
which lays the theoretical foundation for 
data-driven control. 
	Several control problems have been addressed under the new framework, including stabilization and optimization in \cite{van2020data,Wildhagen2021data,rueda2020data,Wang2021data,cortes2021dall}, linear quadratic regulation in \cite{depersis2020Low}, robust control in \cite{depersis2021tradeoffs}, 
	quantized control in \cite{zhao2022data}, 
	model predictive control (MPC) in \cite{berberich2019data,Coulson2019data,Coulson2021Bridging,berberich2020cdc,Liu2021data}, and control of complex networks in \cite{baggio2020Data}.

	Yet, the aforementioned works employ periodic transmission protocols, which may be resource-inefficient for real-world systems in terms of processor usage, communication bandwidth, and energy.
	In cyber-physical networked systems \cite{2013Attack,pang2021anovel} 
	 for instance, whose communication network is shared by many devices, the communication bandwidth is always restricted for each device. 
To tackle this issue, a resource-efficient scheduling approach for data transmissions, known as event-based control, has been widely studied in the context of model-based control.
	
	Generally, there are two event-based approaches that have been proven  effective, namely, event-triggered control and self-triggered control \cite{heemels2012introduction}.
	In the former, an event e.g., a data transmission,
	is triggered only  after some designed triggering condition is met. 
	This condition should be tested at each state or output, thus requiring continuous monitoring of the system.
While a level of robustness against uncertainties and unmodeled dynamics can be expected, having the sensor continuously operating results in waste of resources. 
	Related work can be found in \cite{Lv2020Event,chen2020How,antunes2014rollout,Wakaiki2020event,heemels2013model}.
	Self-triggered control, on the other hand, determines the next sampling time and transmission once a sampled measurement is received, which does not need to continuously sample the outputs; 
	see, e.g., \cite{anta2010to,gommans2015resource,sun2019robust,Almeida2014Self}.
	Notably, the sensors in self-triggered control can be completely shut off between sampling times, which saves energy and prolongs service life of the sensor.
	This feature appears promising in the model-based case and motivates the generalization of self-triggered control from model-based to data-driven settings.
	The key idea of traditional self-triggered control is to predict the future trajectory 
	using the system model. It remains unclear how a trajectory can be obtained in the absence of a model.

	The goal of this paper is to design data-driven self-triggered controllers for unknown linear systems using only some data acquired
 off-line.
	The challenge here lies  in how to obtain a predicted trajectory using only  data. To this end, 
		we begin by considering unknown output feedback systems.
In this setup, a new data-driven MPC scheme is proposed, which generates a sequence of optimal control inputs as well as the associated system trajectory.
	Leveraging the predicted trajectory, a data-driven self-triggering mechanism is designed so that the next triggering time can be dictated without using a system model.  	
We further extend this method to state feedback control systems, where the state is sampled to construct the control input based on a state feedback control law.
Again, a norm minimization problem is formulated to predict the system trajectory using the pre-collected input-state data, whose solution enables design of the data-driven self-triggering mechanism. 

	
	In succinct form, the contribution of this work is threefold, summarized as follows.
	\begin{itemize}
		\item[\textbf{c1)}] 
		To predict the future system trajectory based on input-output data, a data-driven MPC scheme accounting for noisy online outputs is developed for output feedback control systems.
		\item[\textbf{c2)}] 
		Leveraging the solution of the data-driven MPC, a data-driven self-triggering mechanism is designed which determines the next transmission time once a new measurement is received, along with its rigorous and comprehensive stability analysis.
		\item[\textbf{c3)}]
		For state feedback control systems, a norm maximization problem is suggested to predict the system trajectory, whose solution is then
		 employed to construct the data-driven self-triggering law.
	\end{itemize}
	
	\emph{Notation:}
	Denote the set of real numbers (natural numbers) by $\mathbb{R}$ ($\mathbb{N}$), and define $\mathbb{N}_0 := \mathbb{N} \cup \{0\}$.
	For a matrix $M$, if it has full column rank (full row rank), its left pseudo-inverse (right pseudo-inverse) is denoted by $M^\dag$ ($M^\ddag$).
	Given a vector $x\in \mathbb{R}^{n_x}$, $\Vert x\Vert$ is its Euclidean norm, $\Vert x\Vert_{\infty}$ is its infinity norm, and for a positive definite matrix $P = P^T \succ 0$, define the weighted norm $\Vert x \Vert_P = \sqrt{x^\top P x}$.
	Let further $\Vert M\Vert$ be the spectral norm of matrix $M$.
	Let $\underline{\lambda}_{P}$ ($\overline{\lambda}_{P}$) represent the minimum (maximum) eigenvalue of matrix $P$.
	Let $[t_1, t_2]$ denote the time interval from $t_1$ to $t_2$ with discrete times. 
	A function $\alpha : [0,\infty) \rightarrow [0, \infty)$ is said to be of class $\mathcal{K}$ if it is continuous, strictly increasing, and $\alpha(0) = 0$.
	A function $\alpha : [0,\infty) \rightarrow [0, \infty)$ is said to be of class $\mathcal{K}_{\infty}$ if it is of class $\mathcal{K}$ and also unbounded.
	A function $\beta : [0,\infty)\times [0,\infty) \rightarrow [0, \infty)$ is said to be of class $\mathcal{KL}$ if $\beta(\cdot, t)$ is of class $\mathcal{K}$ for each fixed $t \ge 0$ and $\beta(s, t)$ decreases to $0$ as $t \rightarrow \infty$ for each fixed $s \ge 0$.
	
	The Hankel matrix associated with the sequence $\{x_t\}_{t = 0}^{N - 1}$, is denoted by
	\begin{equation}
	H_{L}(x):=\left[
	\begin{matrix}
	x_{0} & x_{1} & \ldots & x_{N-L} \\
	x_{1} & x_{2} & \ldots & x_{N-L+1} \\
	\vdots & \vdots & \ddots & \vdots \\
	x_{L-1} & x_{L} & \ldots & x_{N-1}
	\end{matrix}
	\right].
	\end{equation}
	In this paper, we consider experiments of length $N$, so the dependence of $H_L(x)$ on $N$ will be omitted.
	A stacked window of a sequence, say $\{x_t\}_{t=t_1}^{t_2}$ is given by
	\begin{equation}
	x_{[t_1, t_2]}=\left[
	\begin{matrix}
	x_{t_1} \\
	\vdots \\
	x_{t_2}
	\end{matrix}
	\right].
	\end{equation}

	\section{Preliminaries and Problem Formulation}\label{sec:preliminaries}
	\subsection{Networked control systems}\label{sec:preliminaries:ncs}
	Consider the following discrete-time linear system
	\begin{subequations}\label{eq:sys}
		\begin{align}
			x_{t+1} &= A x_{t} + B u_{t}, \quad t \in \mathbb{N}_0\\
			y_{t} &= Cx_t + Du_t
		\end{align}
	\end{subequations}
	where $x_t\! \in\! \mathbb{R}^{n_x}$, $u_t \!\in\! \mathbb{R}^{n_u}$, and $y_t \!\in\! \mathbb{R}^{n_y}$ are the state, control input, and  output, respectively.
	In this paper, we make the following standing assumptions on the system \eqref{eq:sys}. 
	\begin{assumption}[\emph{Controllability and observability}]\label{as:ctrl}
		The pair $(A, B)$ is controllable, and  $(C, A)$ is observable.
	\end{assumption}
	\begin{assumption}[\emph{Unknown system model}]\label{as:sys}
		The system matrices $(A, B, C, D)$ in \eqref{eq:sys} are unknown.
		Instead, some pre-collected input-output data, i.e., $\{u^p_t, y^p_t\}_{t = 0}^{N - 1}$ are available.
	\end{assumption}
	
	Regarding the two assumptions, a remark comes ready. 
	\begin{remark}
		Assumption \ref{as:ctrl} is  standard 
	for stability analysis of linear systems. 
		As a matter of fact, the data-driven self-triggering controllers along with associated theoretical results developed in this paper can be generalized to linear systems which are stablizable and detectable. In real-world setups, it is almost always impossible to have perfect knowledge of the system matrices $(A, B, C, D)$ due to modeling and/or instrumentation inaccuracies. On the other hand, acquiring input-output data pairs $\{u^p_t, y^p_t\}_{t = 0}^{N - 1}$ by exciting the system with some control inputs is often practical and doable.
	\end{remark}
	
	We further recall the definition of observability index \cite{OreillyJ}. 
	
	\begin{definition}[\emph{Observability index}]\label{def:eta}
		The observability index of linear systems as in \eqref{eq:sys} is defined to be the smallest integer $\eta \in \{1, \cdots, n_x\}$ such that the observability matrix $\Theta$ has full rank $n_x$, i.e.,
		\begin{equation}\label{eq:Theta}
			{\rm rank}(\Theta) \stackrel{\triangle}{=} {\rm rank}\big([C^\top\, (CA)^\top\, \cdots \,(CA^{\eta - 1})^\top]^\top\big) = n_x.
		\end{equation}
	\end{definition}

	\begin{figure}
	\centering
	\includegraphics[width=8.5cm]{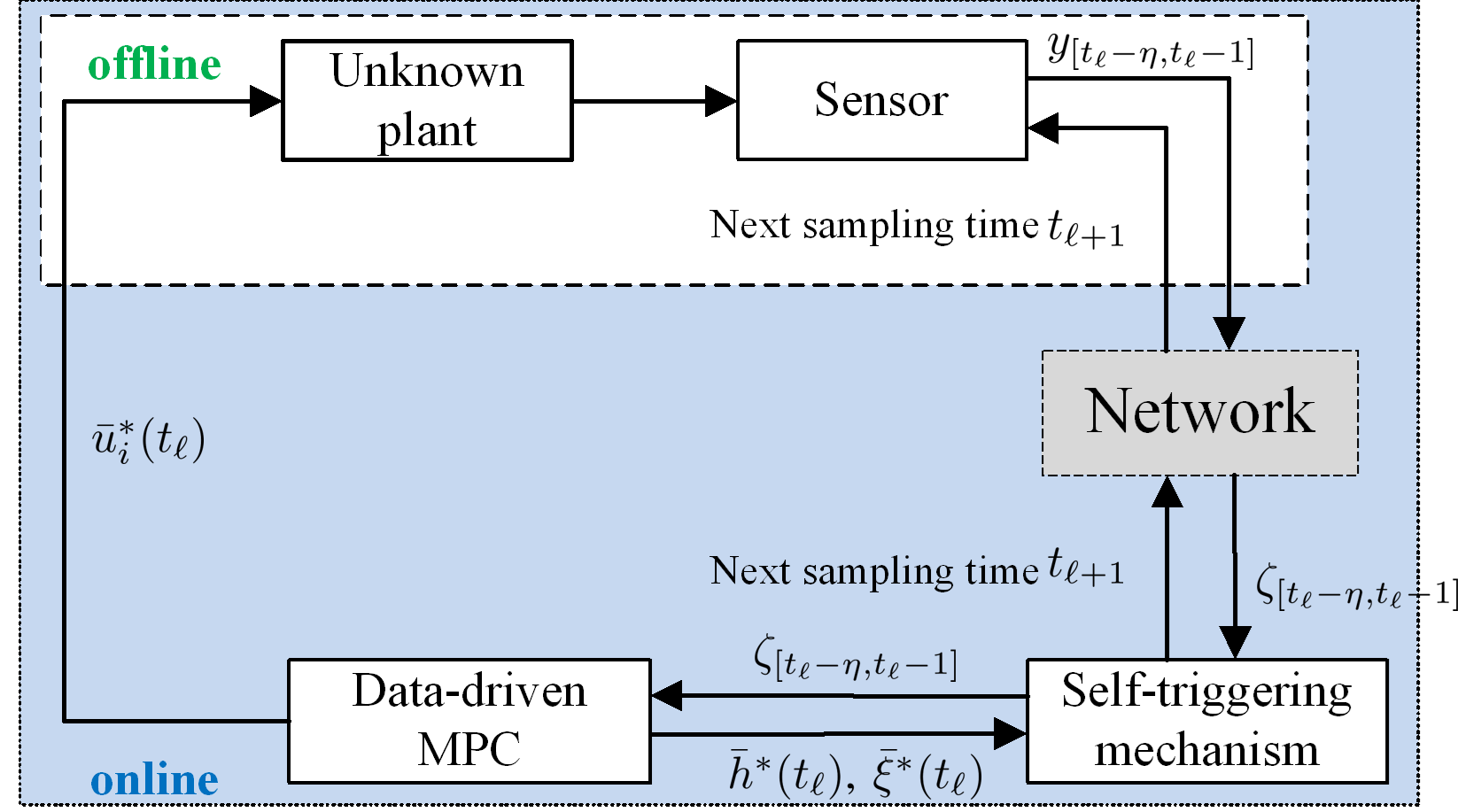}\\
	\caption{Pictorial description of system \eqref{eq:sys} with data-driven self-triggered MPC.}\label{fig:outputsys}
	\centering
\end{figure}
	
	In networked control systems, the plant is connected with a sensor which takes the output measurement at sampling times. 
	Sampled measurements are transmitted to a remote controller to construct control inputs. Frequent sampling and transmission inevitably results in waste of energy and network resources. To save the network bandwidth as well as prolong the life-cycle of the sensor, we incorporate a self-triggering module running at the controller side, to dictate when a new output measurement should be sampled and transmitted. 
	
	Specifically, we consider that 
	the output is transmitted over a network of limited bandwidth to the remote controller, while the controller-to-plant channel is assumed perfect either over wired lines or networks with sufficient bandwidth. Furthermore, to model the network effect,  
the output is corrupted by bounded additive noise $\Vert n_t\Vert\le \bar{n}$ for all $t \in \mathbb{N}_0
$ when transmitted over the network, i.e.,
\begin{equation}
	\zeta_t = y_t + n_t.
\end{equation}
See Fig. \ref{fig:outputsys} for a pictorial depiction of the system structure.	

Let $t_\ell\in\mathbb{N}_0$ stand for the time at which the $\ell$-th event (i.e., sampling and transmission) is triggered by the self-triggering module.
	Assume the sensor is equipped with a buffer of size $\eta$, which records the most recent $\eta$ historical measurements. Once an event is triggered, the data packet consisting of the $\eta$ buffered measurements is sent to the controller. The self-triggering module  computes the next triggering time by evaluating a triggering function on a predicted trajectory along with the received data packet. For example, at time $t_{\ell}$, the $\ell$-th transmission has just been made, and our self-triggering module receives the data packet $y_{[t_\ell - \eta, t_\ell - 1]}:=\{y_{t_\ell-\eta},y_{t_\ell-\eta+1},\cdots\!,y_{t_\ell-1}\}$, and computes the next triggering time $t_{\ell + 1}$ which is fed back to the sensor. During the interval $[t_{\ell - 1},t_{\ell})$, it is only required that the sensor samples output measurements for the sub-interval $[t_\ell - \eta,t_\ell)$, which is often much smaller than $[t_{\ell-1},t_{\ell})$, therefore  {saving considerable energy} as well as expanding life expectancy of the sensor.
	Of course, if the inter-triggering time $t_\ell-t_{\ell-1}$ is smaller than $\eta$, i.e., $t_{\ell  } - t_{\ell-1} \le \eta$, one only needs to send the packet of the new $t_\ell-t_{\ell-1}$ measurements as the old $\eta-(t_\ell-t_{\ell-1})$ ones have already been sent to the controller in the last packet.

	\subsection{Fundamental Lemma}\label{sec:fundamental}
	In this section, we briefly review the data-driven system representation based on the so-called Fundamental Lemma \cite{willems2005note}, which is key to derive our data-driven self-triggering predictive controllers.
	Before presenting the lemma, the definition of persistence of excitation is introduced first.
	
	\begin{definition}[\emph{Persistence of excitation}]\label{def:pe}
		{A sequence  $\{u_t\in\mathbb{R}^{n_u} \}_{t=0}^{N-1}$ is persistently exciting of order $L$ 
			 if ${\rm {rank}}(H_L(u)) = n_u L$.}
	\end{definition}

	Based on Def. \ref{def:pe}, it has been shown in \cite{willems2005note} that any trajectory of system \eqref{eq:sys} can be expressed as a linear combination of pre-collected input-output data $\{u^p_t,y^p_t\}_{t = 0}^{N - 1}$, provided that the input sequence $\{u^p_t\}_{t = 0}^{N - 1}$ is persistently exciting.
	This result, formally summarized below, is also known as the Fundamental Lemma \cite{willems2005note}.
	
	\begin{lemma}[\emph{Fundamental Lemma}]\label{lem:fundamental}
	Consider a persistently exciting input sequence $\{u^p_t\}_{t = 0}^{N - 1}$ of order $L + n_x$.
		A trajectory $\{\bar{u}_t, \bar{y}_t\}_{t = 0}^{L - 1}$ is generated by system \eqref{eq:sys} if and only if there exists a vector $g \in \mathbb{R}^{N - L + 1}$ such that the following holds
		\begin{equation}\label{eq:fundamental}
			\left[
			\begin{matrix}
				H_{L}\!\left(u^p\right) \\
				H_{L}\!\left(y^p\right)
			\end{matrix}
			\right] g=\left[
			\begin{matrix}
				\bar{u}_{[0,L - 1]} \\
				\bar{y}_{[0,L - 1]}
			\end{matrix}
			\right].
		\end{equation}
	\end{lemma}

	It can be concluded from Lemma \ref{lem:fundamental} that to generate a predicted trajectory of length $L$, the pre-collected data should be persistently exciting of at least order $L$.
	To validate this requirement, the next assumption 
	 is imposed.
	\begin{assumption}[{\emph{Pre-collected data}}]\label{as:multitre}
		Let $\{u^p_t\}_{t = 0}^{N - 1}$ be any sequence persistently exciting of order $L + n_x + \eta$.
		The output sequence  $\{y^p_{t}\}_{i = 0}^{N - 1}$ is generated from system \eqref{eq:sys} offline with any initial condition $x^{p}_0$, and input sequence $\{u^p_t\}_{t = 0}^{N - 1}$.
	\end{assumption}

	For notational brevity, we use $u^p$ and $y^p$ to represent sequences $\{u^p_t\}_{t = 0}^{N - 1}$ and $\{y^p_t\}_{t = 0}^{N - 1}$ in the following, respectively.
	Leveraging the observability index $\eta$ in Def. \ref{def:eta}, we construct an extended state for $t\ge \eta$ as follows
	\begin{equation}\label{eq:extx}
		\xi_t := \left[
		\begin{matrix}
			u_{[t - \eta, t - 1]}\\
			y_{[t - \eta, t - 1]}
		\end{matrix}
		\right]\in \mathbb{R}^{n_\xi}
	\end{equation}
	with $n_\xi := (n_u + n_y)\eta$.
	Then, {similar to \cite{1991Linear}}, system \eqref{eq:sys} can be converted into the ensuing linear system
	\begin{subequations}\label{eq:extsys}
		\begin{align}
			\xi_{t+1} &= \tilde{A} \xi_{t} + \tilde{B} u_{t}, \quad t \in \mathbb{N}_\eta\\
			y_{t} &= \tilde{C}\xi_t + \tilde{D} u_t
		\end{align}
	\end{subequations}
	for suitable matrices $(\tilde{A},\tilde{B},\tilde{C},\tilde{D})$ depending on $(A,B,C,D)$. 
	For subsequent analysis, let us denote the equilibrium point of the new system \eqref{eq:extsys} by 
	\begin{equation}
		\xi^e := \left[
		\begin{matrix}
			u_{[0, \eta- 1]}\\
			y_{[0, \eta- 1]}
		\end{matrix}
		\right]
	\end{equation} 
	in which $u_i = u^e$ and $y_i = y^e$ for all $i = 0,1, \cdots, \eta - 1$.
	As a matter of fact, such an augmented system (state) has been commonly employed in studies of data-driven MPC to simplify the stability analysis; see e.g., \cite{berberich2019data, berberich2021mpctec}.

	\section{Data-Driven self-triggering MPC}\label{sec:outputddstmpc}
	
	This section advocates a data-driven approach to designing a controller and a self-triggering mechanism
	 for unknown linear systems as in \eqref{eq:sys}.
	Commonly, a self-triggering mechanism determines the next transmission by comparing the current output with a predicted future output, and triggers a transmission if they differ considerably. However, when the system matrices are unknown, 
	three grand challenges are there:
	c1) how to obtain a predicted system trajectory for the self-triggering mechanism using only pre-collected data?
	c2) how to design stablizing data-driven control inputs? and,
	c3) how to perform performance analysis for the  data-driven self-triggered closed-loop control system? 
We first address c1) and c2) by putting forward a data-driven MPC scheme, which can compute a window of control inputs and {its corresponding predicted outputs} once a new output measurement is transmitted and received.

	\subsection{Data-driven MPC scheme}
	
	Although Lemma \ref{lem:fundamental} characterizes system trajectories in some sense, the correspondence between the input-output trajectory and the coefficient vector is not unique in general.  
		 In addition, the Fundamental Lemma does not account for any noise, so it cannot be directly employed for control and trajectory prediction.
Building on existing works \cite{berberich2019data,Liu2021data}, we propose a new data-driven MPC that, at each triggered time $t_{\ell}$, computes a window of optimal control inputs and associated outputs, using the received noisy data $\zeta_{[t_\ell - \eta, t_\ell - 1]}$.	 
	The control inputs are subsequently employed one by one at times $t \in [t_\ell, t_{\ell + 1})$.
	Let vectors $\bar{u}(t_\ell) := [\bar{u}^\top_{-\eta}(t_\ell)\,\cdots\, \bar{u}^\top_{L - 1}(t_\ell)]^\top$ and $\bar{y}(t_\ell) := [\bar{y}^\top_{-\eta}(t_\ell)\,\cdots\, \bar{y}^\top_{L - 1}(t_\ell)]^\top$ denote the predicted input and output from time $t_\ell - \eta$ to $t_\ell + L - 1$, and  $\bar{\xi}(t_\ell) = [\bar{\xi}_0^\top(t_\ell) \cdots \bar{\xi}_L^\top(t_\ell)]^\top$ the corresponding extended state with
	\begin{align*}
	\bar{\xi}_i(t_\ell) = \left[
	\begin{matrix}
	\bar{u}_{[i-\eta, i-1]}(t_\ell)\\
	\bar{y}_{[i-\eta, i-1]}(t_\ell)
	\end{matrix}
	\right]
	\end{align*}
	where $i \in \{0,1, \cdots\!, L \}$.
	Mathematically,  the following optimization problem is solved at each self-triggered time $t_\ell$
	\begin{subequations}\label{eq:ddmpc}
		\begin{align}
			J^*_L(u_{t_\ell},&  y_{t_\ell}) :=\nonumber 	\\
			\underset{\{g(t_\ell), h(t_\ell)\atop
				\bar{y}(t_\ell), \bar{u}(t_\ell)\}}{\min}
			&\sum_{i = 0}^{L - 1} \Vert \bar{u}_i(t_{\ell}) - u^e\Vert^2_{R}  + \Vert \bar{y}_i(t_{\ell}) - y^e\Vert^2_{Q} + \frac{\lambda_h}{\bar{n}}\Vert h(t_\ell) \Vert^2 \nonumber\\
			&+ \lambda_g \bar{n}\Vert g(t_\ell) \Vert^2  + \Vert \bar{\xi}_L(t_{\ell}) - \xi^e\Vert^2_{P} \label{eq:ddmpc0}\\
			{\rm s.t.}\quad 
			&\left[
			\begin{matrix}
				\bar{u}(t_\ell)\\
				\bar{y}(t_\ell) + h(t_\ell)
			\end{matrix}
			\right] = 
			\left[
			\begin{matrix}
				H_{L + \eta}(u^p)\\
				H_{L + \eta}(y^p)
			\end{matrix}
			\right] g(t_\ell)\label{eq:ddmpc1}\\
			&	\left[
			\begin{matrix}
				\bar{u}_{[-\eta, -1]}(t_\ell)\\
				\bar{y}_{[-\eta, -1]}(t_\ell)
			\end{matrix}
			\right] = 
			\left[
			\begin{matrix}
				u_{[t_\ell - \eta, t_\ell - 1]}\\
				\zeta_{[t_\ell - \eta, t_\ell - 1]}
			\end{matrix}
			\right] \label{eq:ddmpc2}\\
			& \bar{\xi}_L(t_\ell) \in \Xi_\epsilon
			\label{eq:ddmpc3} \\
			& \bar{u}_i \in \mathbb{U},\quad i \in \{1,2,\cdots\!, L - 1\}\label{eq:ddmpc4}.
		\end{align}
	\end{subequations}
	
		The data-driven MPC formulation  \eqref{eq:ddmpc} generalizes that of \cite{berberich2021mpctec} by accommodating noisy outputs. 
		Preselected weighting matrices $R \succ 0$ and $Q\succ 0$ penalize the distances from the predicted input and output to their equilibrium points. 
		Constraint \eqref{eq:ddmpc1} reflects the data-driven system representation in Lemma \ref{lem:fundamental}, where $g(t_\ell) \in \mathbb{R}^{N - L -\eta+ 1}$ stands for the coefficient vector at time $t_{\ell}$, and the vector $h(t_\ell) = [h^\top_{-\eta}(t_\ell)~\cdots~ h^\top_{L - 1}(t_\ell)]^\top$ is added to compensate for the unknown network-induced noise $n_t$ in predicting future outputs.
		Penalties are also imposed on the use of $g(t_\ell)$ and $h(t_\ell)$ in the objective function, {to guarantee system stability and improve robustness} against noise, with weights $\lambda_h > 0$ and $\lambda_g > 0$ balancing between minimizing different costs. 
		Feasibility sets $\mathbb{U}$ and $\Xi_\epsilon$ are prescribed and convex. 
		Finally, \eqref{eq:ddmpc3} is a terminal constraint, which together with the last summand having $P \succ 0$ in \eqref{eq:ddmpc0} ensures that the last predicted extended state $\bar{\xi}_L(t_\ell)$ (at the end of the window) stays in $\Xi_\epsilon$ (i.e., a small neighborhood of the desired equilibrium point). Overall, problem \eqref{eq:ddmpc} is convex and can be efficiently solved using off-the-shelf solvers.

	The terminal constraint \eqref{eq:ddmpc3} and the terminal cost $\Vert \bar{\xi}_L(t_{\ell}) - \xi^e\Vert^2_{P}$ are also critical ingredients in the model-based MPC, see e.g., \cite{JBR-DQM-MMD:19}. 
	To proceed, we make the following assumption on the terminal ingredients $\Xi_{\epsilon}$ {and the corresponding matrix $P$.}

	\begin{assumption}\label{as:terminal}
		There exist matrices $P = P^\top \succ 0$, $K \in \mathbb{R}^{n_u \times n_\xi}$, and a set $\Xi_\epsilon :\{\xi \in \mathbb{R}^{n_\xi} | \Vert \xi - \xi^e\Vert_{P} \le \epsilon \} \subseteq \mathbb{U}^\eta \times \mathbb{R}^\eta$ such that for all $\xi \in \Xi_\epsilon$, $u = u^e + K(\xi - \xi^e)$, and $y = (\tilde{C} + \tilde{D }K)\xi$, the following statements hold true
		\begin{enumerate}
			\item $u \in \mathbb{U}$, $\tilde{A}\xi + \tilde{B} u \in \Xi_\epsilon$ and,
			\item the following inequality holds
			\begin{equation}\label{eq:a+bkxi}
				\Vert (\tilde{A} + \tilde{B} K)\xi\Vert_P^2 \le \Vert \xi\Vert_P^2 - \Vert K\xi\Vert_R^2 -\Vert y\Vert_Q^2.
			\end{equation}
		\end{enumerate}
	\end{assumption}
	
	A feasible data-based  method for selecting the matrix $P$ and the set $\Xi_\epsilon$ has been discussed in \cite[Proposition 10]{berberich2021mpctec}.
	It will be shown in the next subsection that these terminal ingredients have influence on the inter-triggering time, and should be handled with care.
According to \eqref{eq:ddmpc2}, a requirement is imposed on the prediction window as follows.

	\begin{assumption}[{\emph{Prediction horizon}}]\label{as:horizon}
		The predict horizon of the data-driven MPC satisfies $L \ge \eta + 1$.
	\end{assumption}

	Without loss of generality, we consider for simplicity the equilibrium $\xi^e =[(u^e)^\top~ (y^e)^\top]^\top = 0$  in the remainder of the paper. 
	Let $(\bar{y}^*(t_\ell),\bar{u}^*(t_\ell),g^*(t_\ell),h^*(t_\ell))$ denote any optimal solution of problem \eqref{eq:ddmpc} at $t_\ell$, thereby yielding
the predicted extended state 
		 $\bar{\xi}^*(t_\ell)$. 
Let  $u_{t_\ell}$, $y_{t_\ell}$ and $\xi_{t_\ell}$ represent the actual input, output and extended state, respectively.
	The inter-triggering time between two consecutive self-triggered times is defined as $\tau_\ell : = t_{\ell + 1} - t_\ell$ for $\ell \in \mathbb{N}_{0}$, with $t_{0} = 0$.
	With these definitions, the following lemma indicates that the error between the predicted extended state $\bar{\xi}^*(t_\ell)$ and the actual extended state $\xi_t$ is bounded, and it can be rigorously quantified using the optimizer of \eqref{eq:ddmpc}.

	\begin{lemma}\label{lem:outputerror}
		Let Assumptions \ref{as:ctrl}---\ref{as:horizon} hold.
		Consider system \eqref{eq:sys} whose control input is generated by solving problem \eqref{eq:ddmpc} at each self-triggered time $t_\ell$.
		For every $\ell \in \mathbb{N}_0$, the error between $ \xi_{t_{\ell+1}} $ and $\bar{\xi}^*_{\tau_\ell}(t_\ell) $ satisfies
		\begin{subequations}\label{eq:error_xi}
			\begin{align}
				\Vert \xi_e(t_{\ell+1}) \Vert  &:=\Vert \xi_{t_{\ell+1}} - \bar{\xi}^*_{\tau_\ell}(t_\ell) \Vert\\
				&\le \big(\sqrt{\eta}\bar{n} + \Vert h^*_{[-\eta, -1]}(t_\ell)\Vert\big)\sqrt{\sum_{i = \tau_\ell}^{\tau_\ell + \eta - 1} \rho^i} \\
				&~~~+ \Vert h^*_{[\tau_\ell - \eta, \tau_\ell - 1]}(t_\ell)\Vert
			\end{align}	
		\end{subequations}
		where $\rho^i :=\Vert CA^{i + \eta}\Phi^\dag\Vert $ for $i= \tau_\ell,\tau_\ell+1, \cdots, \tau_\ell + \eta - 1$.  
	\end{lemma}

	\begin{proof}
		It follows from \eqref{eq:ddmpc3} that 
		\begin{align*}
		\xi_e(t_\ell + \tau_\ell)& = 
		\left[
		\begin{matrix}
		\bar{u}^*_{[\tau_\ell - \eta, \tau_\ell - 1]}(t_\ell)\\
		y_{[t_\ell + \tau_\ell - \eta, t_\ell + \tau_\ell - 1]}
		\end{matrix}
		\right] \!-\! \left[
		\begin{matrix}
		\bar{u}^*_{[\tau_\ell - \eta, \tau_\ell - 1]}(t_\ell)\\
		\bar{y}^*_{[\tau_\ell - \eta, \tau_\ell - 1]}(t_\ell)
		\end{matrix}
		\right]\\
		& = 
		\left[
		\begin{matrix}
		0\\
		y_{[t_\ell + \tau_\ell - \eta, t_\ell + \tau_\ell - 1]} - \bar{y}^*_{[\tau_\ell - \eta, \tau_\ell - 1]}(t_\ell)
		\end{matrix}
		\right].
		\end{align*}
		Therefore, an upper bound on the extended state error $\Vert \xi_e(t_\ell + \tau_\ell)\Vert$ can be obtained by calculating $y_e(t_\ell + \tau_\ell):= y_{t_\ell + \tau_\ell} - \bar{y}^*_{\tau_\ell}(t_\ell)$.
		It follows from \eqref{eq:ddmpc1} that 
		\begin{align}\label{eq:barystar}
		\bar{y}^*_{\tau_\ell}(t_\ell) = I_{\tau_\ell}H_{\eta + L}(y^p) g^*(t_\ell) - h^*_{\tau_\ell}(t_\ell)
		\end{align}
		where $I_{\tau_\ell}$ denote the corresponding $\tau_\ell$-th block of matrix $H_{\eta + L}(y^p) g^*(t_\ell)$ for the optimizer $\bar{y}^*_{\tau_\ell}(t_\ell)$.
		Let $\hat{y}_{\tau_\ell}(t_\ell) := I_{\tau_\ell}H_{\eta + L}(y^p) g^*(t_\ell)$, which is a trajectory of system \eqref{eq:sys} according to Lemma \ref{lem:fundamental}.
		Based on \eqref{eq:sys}, one can deduce that
		\begin{align*}
		\hat{y}_{\tau_\ell}(t_\ell) \!-\! y_{t_\ell + \tau_\ell} \!= \!CA^{\tau_\ell + \eta}\Theta^{\dag}(\hat{y}_{[-\eta, -1]}(t_\ell) \!-\! y_{[t_\ell - \eta, t_\ell - 1]}).
		\end{align*}
		It follows from \eqref{eq:barystar} and \eqref{eq:ddmpc2} that
		\begin{align*}
		&\hat{y}_{[-\eta, -1]}(t_\ell) - y_{[t_\ell - \eta, t_\ell - 1]}\\
		&~~= y_{[t_\ell - \eta, t_\ell - 1]} + n_{[t_\ell - \eta, t_\ell - 1]} + h^*_{[-\eta, -1]}(t_\ell) - y_{[t_\ell - \eta, t_\ell - 1]}\\
		&~~ = n_{[t_\ell - \eta, t_\ell - 1]} + h^*_{[-\eta, -1]}(t_\ell).
		\end{align*}
		Therefore, the error between the actual output and the predicted output obeys
		\begin{align}\label{eq:ye}
		y_e(t_\ell \!+\! \tau_\ell) \!=\! CA^{\tau_\ell + \eta}\Theta^\dag\big(n_{[t_\ell - \eta, t_\ell - 1]} \!+\! h^*_{[-\eta, -1]}(t_\ell)\big)\!-\! h^*_{\tau_\ell}(t_\ell).
		\end{align}
	We have that
		\begin{align*}
		\Vert \xi_e(t_\ell + \tau_\ell)\Vert &= \Big\Vert\Big[(CA^{\tau_\ell + \eta})^\top~ \cdots~ (CA^{\tau_\ell + 2\eta - 1})^\top\Big]^\top\nonumber\\
		&\quad\times\big[n_{[t_\ell - \eta, t_\ell - 1]} + h^*_{[-\eta, -1]}(t_\ell)\big]\nonumber\\
		& \quad- \big[(h^*_{\tau_\ell}(t_\ell))^\top~ \cdots~ (h^*_{\tau + \eta - 1}(t_\ell))^\top\big]^\top\Big\Vert\\
		&\le \big(\sqrt{\eta}\bar{n} + \Vert h^*_{[-\eta, -1]}(t_\ell)\Vert\big)\sqrt{\sum_{i = \tau_\ell}^{\tau_\ell + \eta - 1} \rho^i}\nonumber \\
		&\quad+ \Vert h^*_{[\tau_\ell - \eta, \tau_\ell - 1]}(t_\ell)\Vert
		\end{align*}
		which completes the proof.
	\end{proof}
According to Lemma \ref{lem:outputerror}, an upper bound on the error between the predicted extended state and the actual extended state is characterized in terms of known or computable parameters $h^*(t_\ell)$, $\bar{n}$, and unknown parameters $\rho^i$ for $i = \tau_\ell, \cdots, \tau_\ell + \eta - 1$.	
Before running the system online, upper bounds of parameters $\rho^i$ can be obtained offline by solving the optimization problem proposed in \cite[Section V. B]{berberich2020cdc} for $i = 1, \cdots, L - 1$
\begin{subequations}\label{eq:rho}
	\begin{align}
	J_i^\prime := & \max_{y,g^\prime}~\Vert y_i \Vert_{\infty}\\
	&\,~~{\rm s.t.}~ \Vert y_{[0, \eta - 1]}\Vert_{\infty} \le 1 \label{eq:rho_1}\\
	&\,\quad\quad\left[
	\begin{matrix}
	H_{i + 1}(u^p)\\
	H_{i + 1}(y^p)
	\end{matrix}
	\right]g^\prime = 
	\left[
	\begin{matrix}
	0\\
	y_{[0,i]}
	\end{matrix}
	\right]\label{eq:rho_2}
	\end{align}
\end{subequations}
where $H_{i + 1}(u^p)$ denotes the first $(i + 1)n_u$ rows of matrix $H_{L + \eta}(u^p)$,  and likewise for $H_{L + \eta}(y^p)$; vector $g\prime \in \mathbb{R}^{N - L - \eta + 1}$, and $y_{[0,i]}$ consists of the  $i + 1$ predicted outputs.	
It has been shown in \cite{berberich2020cdc}  that $J_i^\prime \ge \Vert CA^{i + \eta}\Phi^\dag\Vert$, and hence, an upper bound for $\rho^i$ is obtained, i.e., $J_i^\prime\ge \rho^i$.
Therefore, one gets an upper bound on $\Vert \xi_e(t_{\ell + 1})\Vert$.
	\subsection{Self-triggering mechanism}
	
	Observe that the prediction horizon constrains at most $L - 1$ future outputs at time $t_\ell$ can be obtained by solving \eqref{eq:ddmpc}.
	Therefore, the inter-triggering time between any two consecutive self-triggered times obeys  $\tau_\ell \le L - 1$ for all $\ell \in \mathbb{N}_0$.
Leveraging the upper bound in Lemma \ref{lem:outputerror}, our self-triggering mechanism is given in the following lemma.

	\begin{lemma}\label{lem:tau}
		Let Assumptions \ref{as:ctrl}---\ref{as:horizon} hold. 
		Consider system \eqref{eq:sys} with control inputs generated by the data-driven MPC \eqref{eq:ddmpc} at each self-triggered time $t_\ell$ for all $\ell \in \mathbb{N}_0$.
		Assume further that \eqref{eq:ddmpc} is feasible at $t_0$.
		For appropriate $\lambda_g > 0$, $\lambda_h > 0$, $\epsilon \ge 0$, and $P \succ 0$ satisfying Assumption \ref{as:terminal}, there exists a constant $\bar{\sigma} \in (0,1)$ such that for all $\sigma \in (0,\bar{\sigma})$,  problem \eqref{eq:ddmpc}  is feasible at all $t_{\ell}\in \mathbb{N}_0$, 
			if i) the constant $r$, matrices $P_r$ and $K_r$ are chosen obeying Assumption \ref{as:terminal} and 
			\begin{align}\label{eq:rleepsilon}
			\frac{\bar{\lambda}_{P_r}}{\underline{\lambda}_{P}}\Big(1 - \frac{\underline{\lambda}_{K_r^\top R K_r}}{\bar{\lambda}_{P_r}}\Big)^L r^2 \le \epsilon^2 \le r^2
			\end{align}
			and ii) the inter-triggering time satisfies 
			\begin{align}\label{eq:tau}
			\tau_\ell = \min\big\{ \hat{\tau}_\ell,\, \check{\tau}_\ell,\, L - 1 \big\}
			\end{align}
			with
			\begin{align}
			\hat{\tau}_\ell &:= \sup\bigg\{\tau_\ell\Big| (\sqrt{\eta}\bar{n} + \Vert h^*_{[-\eta, -1]}(t_\ell)\Vert)\sqrt{\sum_{i = \tau_\ell}^{\tau_\ell + \eta - 1} \rho^i} \nonumber\\
			&~\quad \quad\quad+ \Vert h^*_{[\tau_\ell - \eta, \tau_\ell - 1]}(t_\ell)\Vert  \le \frac{r}{\bar{\lambda}_{P_r}} - \Vert \bar{\xi}^*_{\tau_\ell}(t_\ell)\Vert \bigg\}\label{eq:tautau1}
			\end{align}
			and
			\begin{align}
			\check{\tau}_\ell &:=\sup\bigg\{\tau_\ell\Big| 2 \Big[\bar{\lambda}_Q + 2 \bar{\lambda}_{P_r} + \lambda_g \bar{n} \Vert H^\ddag_{u\xi}\Vert^2\Big(1 + \frac{\bar{\lambda}_{P_r}}{\underline{\lambda}_{R}}\Big)\Big]\nonumber\\
			&\times\Big[\big(\eta\bar{n}^2 + \Vert h^*_{[-\eta, -1](t_\ell)} \Vert^2\big)\sum_{i = 0}^{\tau_\ell - 1}(\rho^{i + \eta})^2 + \sum_{i = 0}^{\tau_\ell - 1}\Vert h^*_{i}(t_\ell)\Vert^2\Big]\nonumber\\
			& + \lambda_h\eta\bar{n}^2  + \epsilon^2 + 2\Big[\bar{\lambda}_{P_r} + \lambda_g\bar{n}\Vert H_{u \xi}^\ddag\Vert^2\Big(1 + \frac{\bar{\lambda}_{P_r}}{\underline{\lambda}_{R}}\Big)\Big]\nonumber\\
			&\times \Vert\bar{\xi}^*_{\tau_\ell}(t_\ell)\Vert^2 \le \sigma \sum_{i = 0}^{\tau_\ell - 1} \Vert \bar{\xi}^*_i(t_\ell)\Vert^2  \bigg\}\label{eq:tautau2}
			\end{align}
			where matrix $H_{u\xi}$ is defined by
			\begin{align}\label{eq:Huxi}
			H_{u \xi} = \left[
			\begin{matrix}
			H_{L + \eta}(u^p)\\
			H_1(\xi^p_{N - L - \eta + 1})
			\end{matrix}
			\right].
			\end{align}
	\end{lemma}
\begin{proof}
	First, for some $\ell \in \mathbb{N}_0$, we assume that the problem is feasible at $t_\ell$ and construct a candidate solution at $t_{\ell + 1}$.
	Then, leveraging a carefully designed Lyapunov function, recursive feasibility of problem \eqref{eq:ddmpc} is proven by showing the decrease of the Lyapunov function.
	
	Suppose that \eqref{eq:ddmpc} is feasible at $t_\ell$ for some $\ell \in \mathbb{N}_0$. 
	Denote a candidate solution of  \eqref{eq:ddmpc} at $t_{\ell + 1}$ by $\bar{y}(t_{\ell + 1})$, $\bar{u}(t_{\ell + 1})$, $g(t_{\ell + 1})$, and $h(t_{\ell + 1})$.
	Using \eqref{eq:error_xi}, one gets from \eqref{eq:tau}--\eqref{eq:tautau1} that
	\begin{align*}
	\Vert \xi_{t_{\ell + 1}}\Vert_{P_r} \le \Vert \xi_e(t_{\ell + 1})\Vert_{P_r} + \Vert \bar{\xi}^*_{\tau_\ell}(t_\ell)\Vert_{P_r} \le r.
	\end{align*}
	Thus, for terminal region $
	\Xi_r:\{\xi \in \mathbb{R}^{n_\xi} | \Vert \xi - \xi^e\Vert_{P_r} \le r\}$
	there exist matrices $P_r$ and $K_r$ such that conditions in Assumption  \ref{as:terminal} hold.
	Let $\bar{u}_i(t_{\ell + 1}) = K \bar{\xi}_{i}(t_{\ell + 1})$ and $\bar{y}_i(t_{\ell + 1}) = y_{i + t_{\ell+ 1}}$ for $i \in \{0,\cdots\!, L - 1\}$.
	Capitalizing on \eqref{eq:ddmpc1}, the slack variable $h(t_{\ell + 1})$  can be chosen as $h_{[-\eta, -1]}(t_{\ell + 1}) = n_{[t_{\ell + 1} - \eta, t_{\ell + 1} - 1]}$ and $h_i(t_{\ell + 1}) = 0$ for $i \in \{0, 1,\cdots, L - 1\}$.
	Set 
	\begin{align}\label{eq:barg}
	{g}(t_{\ell + 1}) = H^\ddag_{u\xi}
	\left[
	\begin{matrix}
	u_{[0, L - 1]}(t_{\ell + 1})\\
	\xi(t_{\ell + 1})
	\end{matrix}
	\right]
	\end{align}
	with  $H_{u\xi}$ in \eqref{eq:Huxi}.
	It follows from \eqref{eq:a+bkxi} that
	\begin{align*}
	\Vert \bar{\xi}_{i + 1}(t_{\ell + 1})\Vert^2_{P_r} &\le \Vert \bar{\xi}_i(t_{\ell + 1})\Vert^2_{P_r} - \Vert \bar{\xi}_{i}(t_{\ell + 1})\Vert_{K_r^\top R K_r}^2\\
	& \le \big[1 - \underline{\lambda}_{K_r^\top R K_r}/\bar{\lambda}_{P_r}\big]\Vert \bar{\xi}_{i}(t_{\ell + 1}) \Vert_{P_r}^2
	\end{align*}
	for all $i \in \{0, \cdots, L - 1\}$.
	Recursively, we arrive at
	\begin{align}\label{eq:xil0}
	\Vert \bar{\xi}_L(t_{\ell + 1})\Vert_{P_r}^2 \le \big[1 - \underline{\lambda}_{K_r^\top R K_r}/\bar{\lambda}_{P_r}\big]^L\Vert \bar{\xi}_{0}(t_{\ell + 1}) \Vert_{P_r}^2.
	\end{align}
	Combining \eqref{eq:rleepsilon} with $\Vert \bar{\xi}_{0}(t_{\ell + 1}) \Vert_{P_r}^2 \le r^2$, inequality \eqref{eq:xil0} implies that $\Vert \bar{\xi}_L(t_{\ell + 1})\Vert_{P}^2 \le \epsilon^2$.
	Therefore, if problem \eqref{eq:ddmpc} is feasible at $t_\ell$, it is feasible at $t_{\ell + 1}$.
	
	Next, we construct a Lyapunov function to show the recursive feasibility.
	Since $(C,A)$ is observable, $(\tilde{C}, \tilde{A})$ is detectable.
	According to \cite{CAI2008326}, there exists an input-output-to-state stability Lyapunov function $W(\xi) = \xi^\top P_{\xi}\xi$ such that 
	\begin{equation}\label{eq:iossw}
	W(\tilde{A}\xi + \tilde{B}u) - W(\xi) \le -\frac{1}{2}\Vert \xi \Vert^2 + c_1 \Vert u\Vert^2 + c_2 \Vert y\Vert^2
	\end{equation}
	for suitable $c_1, c_2 > 0$, $P_\xi \succ 0$, and all $u \in \mathbb{R}^{n_u}$, $\xi\in \mathbb{R}^{n_\xi}$, $y = \tilde{C} \xi + \tilde{D}\xi$.
	Consider the Lyapunov function $V(\xi_t) = J_L^*(\xi_t) + \gamma W(\xi_t)$ having $\gamma  = \min\{\underline{\lambda}_Q, \underline{\lambda}_R\} / {\max\{c_1, 2c_2\}}$.
	Recalling the candidate solution $(\bar{u}(t_{\ell + 1}), \bar{y}(t_{\ell + 1}), g(t_{\ell + 1}), h(t_{\ell + 1}))$ for problem \eqref{eq:ddmpc}  in Lemma \ref{lem:tau}. 
	Denote the corresponding cost function by $\bar{J}_L(\xi_{t_\ell})$.
	The optimizer of  \eqref{eq:ddmpc} at $t_{\ell}$ is denoted by $(\bar{u}^*(t_\ell), \ \bar{y}^*(t_\ell), g^*(t_\ell), h^*(t_\ell))$.
	The difference of the Lyapunov function at two consecutive self-triggered times satisfies
	\begin{align}\label{eq:diffV}
	&V(\xi_{t_{\ell + 1}}) - V(\xi_{t_\ell}) \\
	=\;& \bar{J}_L(\xi_{t_{\ell + 1}}) + \gamma W(\xi_{t_{\ell + 1}}) - J^*_L(\xi_{t_\ell}) - \gamma W(\xi_{t_\ell})\\ 
	\le\;& \underbrace{\big(\bar{J}_L(\xi_{t_{\ell + 1}})- J^*_L(\xi_{t_\ell})\big)}_{\stackrel{\triangle}{=}\mathcal{T}_1} + \underbrace{\gamma \big(W(\xi_{t_{\ell + 1}})  -  W(\xi_{t_\ell})\big)}_{\stackrel{\triangle}{=}\mathcal{T}_2}.
	\end{align} 
	
	Next, we derive upper bounds on the two terms.
	
	\emph{Part a: Upper-bounding $\mathcal{T}_1$.}
	According to \eqref{eq:ddmpc0}, it can be obtained that
	\begin{align*}
\mathcal{T}_1	& \le \sum_{i = 0}^{L - 1}\big(\Vert \bar{u}_i(t_{\ell + 1}) \Vert^2_R + \Vert \bar{y}_i(t_{\ell + 1}) \Vert^2_Q\big) \nonumber\\
	&\quad+ \lambda_g\bar{n}\Vert g(t_{\ell + 1}) \Vert^2 + (\lambda_{h}/\bar{n})\Vert h(t_{\ell + 1})\Vert^2 + \Vert \bar{\xi}_L(t_{\ell + 1})\Vert^2_{P}\nonumber\\
	&\quad - \Big[\sum_{i = 0}^{L - 1}\big(\Vert \bar{u}^*_i(t_\ell) \Vert^2_R + \Vert \bar{y}^*_i(t_\ell) \Vert^2_Q\big) + \lambda_g\bar{n}\Vert g^*(t_{\ell + 1}) \Vert^2 \nonumber\\
	&\quad+ (\lambda_{h}/\bar{n})\Vert h^*(t_{\ell + 1})\Vert^2 + \Vert \bar{\xi}^*_L(t_{\ell + 1})\Vert^2_{P}\Big].
	\end{align*}
	In addition, since $\bar{u}_i(t_{\ell + 1}) = K \bar{\xi}(t_{\ell + 1})$, it can be deduced recursively from \eqref{eq:a+bkxi} that
	\begin{align*}
	&\sum_{i = 0}^{L - 1}\big(\Vert \bar{u}_i(t_{\ell + 1}) \Vert^2_R + \Vert \bar{y}_i(t_{\ell + 1}) \Vert^2_Q\big) \le \Vert {\xi}_{t_{\ell + 1}}\Vert_{P_r}^2 \!-\! \Vert \bar{\xi}_L(t_{\ell})\Vert_{P_r}^2\nonumber\\
	& \le 2\bar{\lambda}_{P_r}\big(\Vert \xi_e(t_{\ell + 1})\Vert^2 + \Vert \bar{\xi}^*_{\tau_\ell}(t_{\ell})\Vert^2\big).
	\end{align*}
	Notice from \eqref{eq:barg} that
	\begin{align*}
	&\quad~\Vert g(t_{\ell + 1})\Vert^2 \le \Vert H^{\ddag}_{u \xi}\Vert^2(\Vert \bar{u}_{[0, L - 1](t_\ell)}\Vert^2 + \Vert \xi_{t_{\ell + 1}}\Vert^2) \\
	&\le \Vert H^{\ddag}_{u \xi}\Vert^2(1 + \bar{\lambda}_{P_r}/\underline{\lambda}_{R})\Vert \xi_{t_{\ell + 1}}\Vert^2\\
	& \le 2 \Vert H^{\ddag}_{u \xi}\Vert^2\Big(1 + \frac{\bar{\lambda}_{P_r}}{\underline{\lambda}_{R}}\Big)\big(\Vert \xi_e(t_{\ell + 1})\Vert^2 + \Vert \bar{\xi}^*_{\tau_\ell}(t_{\ell})\Vert^2\big).
	\end{align*} 
	Since $h_{[-\eta, -1]}(t_{\ell + 1}) = n_{[t_{\ell + 1} - \eta, t_{\ell + 1} - 1]}$ and $h_i(t_{\ell + 1}) = 0$ for $i \in \{0, \cdots, L - 1\}$, one has that
	\begin{equation*}
	(\lambda_h/\bar{n})\Vert h(t_{\ell + 1})\Vert^2 \le \lambda_h\eta \bar{n}.
	\end{equation*}
	
	{\it Part b: Upper-bounding $\mathcal{T}_2$.}
	It follows recursively from \eqref{eq:iossw} that 
	\begin{align}\label{eq:ww}
	&\quad W(\xi_{t_{\ell + 1}}) - W(\xi_{t_\ell})\nonumber\\
	& \le -\frac{1}{2} \Vert \xi_{[t_{\ell}, t_{\ell + 1} - 1]}\Vert^2 + c_1 \Vert \bar{u}^*_{[0, \tau_{\ell} - 1]}(t_{\ell})\Vert^2 + c_2\Vert y_{[t_\ell, t_{\ell + 1}]}\Vert^2\nonumber\\
	& \le -\frac{1}{2} \Vert \xi_{[t_{\ell}, t_{\ell + 1} - 1]}\Vert^2 \!+\! c_1 \Vert \bar{u}^*_{[0, \tau_{\ell} - 1]}(t_{\ell})\Vert^2\!+\! 2c_2\Vert \bar{y}^*_{[0, \tau_\ell - 1]}(t_\ell)\Vert^2\nonumber\\
	&\quad  + 2c_2\Vert y_{[t_\ell, t_{\ell + 1}]} - \bar{y}^*_{[0, \tau_\ell - 1]}(t_\ell)\Vert^2. 
	\end{align}
	
	Since $\gamma  = {\min\{\underline{\lambda}_Q, \underline{\lambda}_R\}}/{\max\{c_1, 2c_2\}}$, we have that
	\begin{align}\label{eq:gammauy}
	&\quad \gamma(c_1 \Vert \bar{u}^*_{[0, \tau_{\ell} - 1]}(t_{\ell})\Vert^2\!+\! 2c_2\Vert \bar{y}^*_{[0, \tau_\ell - 1]}(t_\ell)\Vert^2)\nonumber\\
	&\le \sum_{i = 0}^{\tau_{\ell} - 1} \big(\Vert\bar{u}_i^*(t_\ell)\Vert_R^2 + \Vert\bar{y}_i^*(t_\ell)\Vert_Q^2\big).
	\end{align}
	
	Substituting \eqref{eq:ye} and \eqref{eq:gammauy} into \eqref{eq:ww}, we arrive at
	\begin{align}
	 \mathcal{T}_2 &\le -\frac{\gamma}{2}\Vert \xi_{[t_\ell, t_{\ell + 1}  - 1]}\Vert^2+ 2\bar{\lambda}_Q\sum_{i = 0}^{\tau_{\ell} - 1}\Vert h^*_i(t_\ell)\Vert^2 \nonumber\\
	&\quad+ \sum_{i = 0}^{\tau_{\ell} - 1} \big(\Vert\bar{u}_i^*(t_\ell)\Vert_R^2 + \Vert\bar{y}_i^*(t_\ell)\Vert_Q^2\big) \nonumber\\
	&\quad + \bar{\lambda}_Q\big(2 \eta\bar{n}^2 + 2 \Vert h^*_{[-\eta, -1]}(t_\ell)\Vert^2\big)\sum_{i = 0}^{\tau_\ell - 1}(\rho^{i + \eta})^2 .
	\end{align}

	
	Combining the results in {Parts a} and {b}, it holds that
	\begin{align}\label{eq:Verror}
	& V(\xi_{t_{\ell + 1}}) - V(\xi_{t_\ell}) \le -\frac{\gamma}{2}\Vert \xi_{[t_{\ell }, t_{\ell + 1} - 1]}\Vert^2 + 2\bigg[\bar{\lambda}_Q + 2\bar{\lambda}_{P_r} \nonumber\\
	&+ \lambda_g\bar{n} \Vert H_{u \xi}^\ddag\Vert^2\Big(1 + \frac{\bar{\lambda}_{P_r}}{\underline{\lambda}_{R}}\Big)\bigg]\bigg[\big(\eta\bar{n}^2 + \Vert h^*_{[-\eta, -1]}(t_{\ell})\Vert^2\big)\nonumber\\
	&\times  \sum_{i = 0}^{\tau_\ell - 1}(\rho^{\eta + i})^2 + \sum_{i = 0}^{\tau_{\ell} - 1}\Vert h^*_{i}(t_\ell)\Vert^2\bigg] + \lambda_h \eta\bar{n}^2 + 2\bigg[\bar{\lambda}_{P_r} \nonumber\\
	& + \lambda_g\bar{n} \Vert H_{u \xi}^\ddag\Vert^2\Big(1 + \frac{\bar{\lambda}_{P_r}}{\underline{\lambda}_{R}}\Big)\bigg]\Vert \bar{\xi}^*_{\tau_\ell}(t_\ell)\Vert^2.
	\end{align}
	In addition, \eqref{eq:tau} indicates that $\tau_\ell \le \check{\tau}_\ell$, and hence \eqref{eq:Verror} obeys
	\begin{align*}
	&\quad V(\xi_{t_{\ell + 1}}) - V(\xi_{t_\ell}) \\
	&\le -\Big(\frac{\gamma}{2} \!-\! \sigma\Big) \Vert \xi_{t_\ell}\Vert^2+ 2\sigma\eta\bar{n}^2 \!+\!\!\! \sum_{i = 1}^{\tau_{\ell} - 1}\!\!\Big(\sigma \Vert \bar{\xi}_{i}(t_\ell)\Vert^2 \!-\! \frac{\gamma}{2} \Vert\xi_{t_\ell + i}\Vert^2\Big)\nonumber\\
	& \le -\Big(\frac{\gamma}{2} - \sigma\Big) \Vert \xi_{t_\ell}\Vert^2 + 2\sigma\eta\bar{n}^2 +\frac{\gamma}{2}\sum_{i = 1}^{\tau_{\ell} - 1} \Vert \bar{\xi}_{e}(t_\ell + i)\Vert^2\nonumber\\
	& \le -(\gamma/2 - 2\sigma) \Vert \xi_{t_\ell}\Vert^2 + \alpha_1(\bar{n})
	\end{align*}
	where $\alpha_1(\bar{n}) := \bar{n}[J^*_L(t_{\ell})/\lambda_h + \sum_{i = 1}^{L - 1}(\rho^{\eta+i})(\eta + J^*_L(t_{\ell})/\lambda_h)] + 2\sigma\eta\bar{n}^2$ is a $\mathcal{K}_{\infty}$ function, and  $\bar{\sigma} \le \gamma /4$.
	Moreover, according to \cite[Lemma 1]{berberich2019data},  the Lyapunov function satisfies $\underline{\lambda}_{P_{\xi}}\Vert \xi_{t_\ell} \Vert^2 \le V(\xi_{t_\ell}) \le \gamma_1 \Vert \xi_{t_\ell} \Vert^2 + \alpha_2(\bar{n})$ for every $\ell \in \mathbb{N}_0$ and some constant $\gamma_1 >0$ with function $\alpha_2(\cdot) \in \mathcal{K}_{\infty}$.
	Hence, we have that
	\begin{align*}
	V(\xi_{t_{\ell + 1}})& \le \underbrace{\Big(1-\frac{\gamma- 4\sigma}{2\gamma_1}\Big)}_{\stackrel{\triangle}{=}\gamma_2<1}V(\xi_{t_\ell}) + \underbrace{\frac{\gamma- 4\sigma}{2\gamma_1}\alpha_2(\bar{n}) + \alpha_1(\bar{n})}_{\stackrel{\triangle}{=}\alpha_3(\bar{n}) \in \mathcal{K}_{\infty}}.
	\end{align*}
	In conclusion, the Lyapunov function decreases at every self-triggered time. Since problem \eqref{eq:ddmpc} is feasible at $t_0$, it is feasible at $t_{\ell}$ for all $\ell \in \mathbb{N}_0$.
\end{proof}

	\begin{remark}
			Condition \eqref{eq:tautau1} ensures that problem \eqref{eq:ddmpc} is feasible at $t_{\ell + 1}$ if it is feasible at $t_\ell$, which is also known as \emph{recursive feasibility}.
			In addition, condition \eqref{eq:tautau2} guarantees the Lyapunov function of system \eqref{eq:sys} decreases along the self-triggered times.
			Therefore, the self-triggering law reduces the communication load while ensuring the recursive feasibility of  \eqref{eq:ddmpc}. 
			{Similar to model-based self-triggered control,
			the choice of the contraction factor $\sigma \in (0,1)$ determines the trade-off between having a better control performance and incurring less transmissions.}
			The smaller $\sigma$, the better the system performance. 
	\end{remark}
	
	Merging the data-driven MPC and self-triggering schemes in Lemma \ref{lem:tau}, our data-driven self-triggered controller is presented in Alg. \ref{alg:ddmpc}, with stability guarantees provided below.
	\begin{algorithm}[h]
		\caption{Data-driven Self-triggered Predictive Control.}
		\label{alg:ddmpc}
		\begin{algorithmic}[1]
			\STATE {\bfseries Input:} 
			Prediction window $L \ge  {\eta} + 1$; coefficients $R \succ 0$, $Q \succ 0$, $P \succ 0$, $\lambda_g > 0$, $\lambda_h > 0$, $\sigma \in (0, 1)$, $\epsilon > 0$ and $r > 0$ satisfying \eqref{eq:rleepsilon}; noise bound $\bar{n}$;
			data $\{u^p_{[0, N - 1]}, y^p_{[0, N - 1]}\}$ of \eqref{eq:sys} from initial condition $x_0^{p}$, 
			with 
			 $u^p_{[0, N - 1]}$  persistently exciting of order $L + \eta + 1$.
			\STATE {\bfseries Construct} Hankel matrices for the input,  output, and the extended state trajectories, i.e., $H_{L + \eta}(u^p)$, $H_{L + \eta}(y^p)$, and $H_{u\xi}:=[H_{L + \eta}^\top(u^p)~ H_1^\top(\xi^p_{N - L - \eta + 1})]^\top$.
			\STATE {\bfseries Compute} $\rho^i$ for $i \in \{1,2, \ldots, L - 1\}$ from \eqref{eq:rho}.
			\STATE {\bfseries If} $t = t_\ell$, do \label{alg:ddmpc1}
			\STATE 
			\begin{itemize}
				\item [] Use the past $\eta$ measurements, i.e.,  $u_{[t - \eta, t - 1]}$ and $\zeta_{[t - \eta, t - 1]}$, to solve problem \eqref{eq:ddmpc}.\\
				Determine the next self-triggered time $t_{\ell + 1}$ such that the inter-triggering time obeys \eqref{eq:tau}--\eqref{eq:tautau2}.\\
				Set $u_t = \bar{u}^*_0(t)$.\\
				Set $t_{\ell} = t_{\ell + 1}$
			\end{itemize}
			\STATE {\bfseries Else if} $t \ne t_\ell$
			\STATE \quad \quad Set $u_t = \bar{u}^*_{t - t_\ell}(t_\ell)$.\label{alg:ddmpcu1}
			\STATE {\bfseries End if}
			\STATE {\bfseries Set} $t = t + 1$ and go back to \ref{alg:ddmpc1}.
		\end{algorithmic}
	\end{algorithm}
	
	\begin{theorem}
		Let Assumptions \ref{as:ctrl}---\ref{as:horizon} hold.
		Consider system \eqref{eq:sys} whose control input is generated by solving problem \eqref{eq:ddmpc} at each self-triggered time $t_\ell$.
		If the sequence $\{t_\ell\}_{\ell \in \mathbb{N}_0}$ satisfies the conditions in Lemma \ref{lem:tau}, then system \eqref{eq:sys} achieves practical exponential stability under the controller \eqref{eq:ddmpc}.
	\end{theorem}
	
	\begin{proof} 
		Recall from Lemma \ref{lem:tau} that it holds for any $t_\ell$
			\begin{equation}\label{eq:invariantV}
			V(\xi_{t_{\ell + 1}}) \le \gamma_2 V(\xi_{t_\ell}) + \alpha_3(\bar{n}).
			\end{equation}
			Recursively, it can be deduced that
			\begin{equation}\label{eq:Vtell+1V0}
			V(\xi_{t_\ell}) \le \gamma_2^{\ell}V(\xi_{t_0}) + \alpha_3(\bar{n})\sum_{i = 0}^{\ell - 1}\gamma_2^i.
			\end{equation} 
			In addition, noticing that $\underline{\lambda}_{P_{\xi}}\Vert \xi_{t_\ell}\Vert^2 \le V(\xi_{t_\ell}) \le \gamma_1 \Vert \xi_{t_\ell} \Vert^2 + \alpha_2(\bar{n})$, we arrive at
			\begin{equation*}
			\Vert \xi_{t_{\ell}}\Vert^2 \le \frac{\gamma_1}{\underline{\lambda}_{P_{\xi}}}\gamma_2^\ell\Vert \xi_{t_0}\Vert^2 + \Big(\alpha_2(\bar{n})\gamma_2^\ell + \alpha_3(\bar{n})\sum_{i = 0}^{\ell - 1}\gamma_2^i\Big)
			\end{equation*}
			which verifies the system's practical exponential stability.
	\end{proof}
	
	\section{Data-driven Self-triggered Predictive Control for State Feedback Systems}\label{sec:stateddmpc}
	{In this section, we consider a special case of the system in \eqref{eq:sys} having  $C = I$ and $D = 0$; that is, when input-state data become available.
	In this setup, a state feedback controller is employed.
	Although simpler than the MPC scheme, the state feedback controller does not provide a predicted system trajectory, which makes design of the self-triggering mechanism difficult.
	To address this issue, a data-based norm maximization problem is put forth such that a future system trajectory can be predicted from the input-state data.
	The key idea behind constructing a predicted trajectory for our self-triggered state feedback controller is combining the results of Lemmas \ref{lem:fundamental} and  \ref{lem:outputerror}. 
Therefore, the trajectory prediction along with the self-triggering mechanism here builds on the idea of \eqref{eq:ddmpc}.
}

%
	Consider the following state feedback controller
	\begin{subequations}\label{eq:statesys}
		\begin{align}
		x_{t+1} &= A x_{t} + B u_{t}, \quad &t& \in \mathbb{N}_0\\
		u_{t} &= K {\zeta}_{t_{\ell}}, \quad &t_\ell& \le t < t_{\ell + 1}\\
		\zeta_{t_\ell} &= x_{t_\ell} + n_{t_\ell}
		\end{align}
	\end{subequations}
	where the gain matrix $K \in \mathbb{R}^{n_u \times n_x}$ can be obtained via a data-based linear matrix inequality method in e.g. \cite{persis2020data}.
	At each self-triggered time $t_{\ell}$, the state $x_{t_\ell}$ is sampled and sent to the controller.
	Similar as in Section \ref{sec:outputddstmpc}, the transmission is corrupted by noise $n_t$, and hence the controller receives $\zeta_{t_\ell}$.
	The control input $u_{t_\ell}$ is sent to the plant, and kept the same within the interval $[t_\ell, t_{\ell + 1})$. See Fig. \ref{fig:statesys} for an illustration.
	\begin{figure}
		\centering
		\includegraphics[width=8cm]{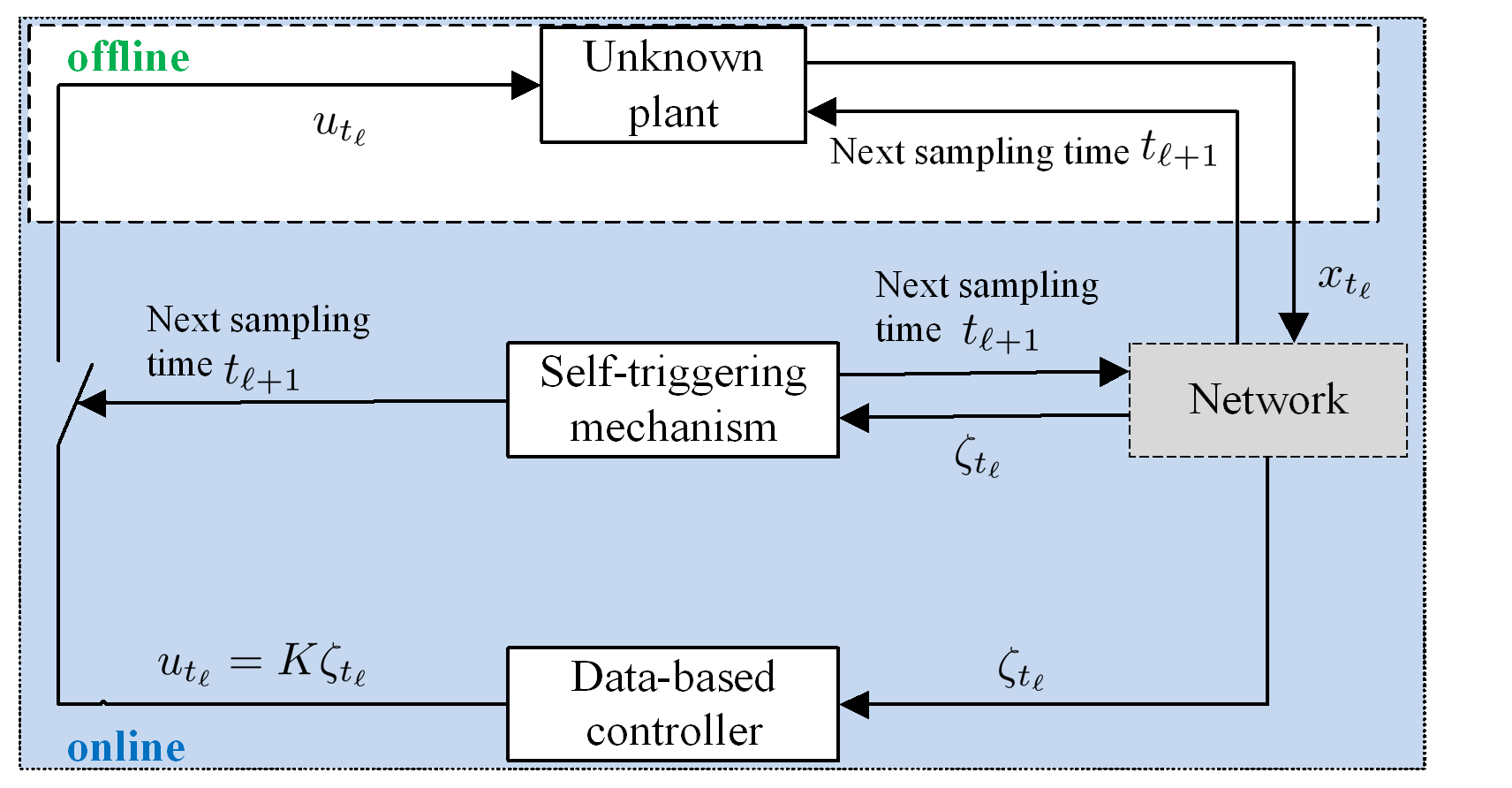}\\
		\caption{Data-driven self-triggered state feedback controller for \eqref{eq:statesys}.
		}\label{fig:statesys}
		\centering
	\end{figure}	
Instead of input-output data in Section \ref{sec:outputddstmpc}, input-state data are available in this section, and hence Assumptions \ref{as:sys} and \ref{as:multitre} are replaced as follows.
	\begin{assumption}[\emph{Unknown system model}]\label{as:statesys}
	The system matrices $(A, B)$ in \eqref{eq:statesys} are unknown, but some pre-collected input-state data $\{u^p_t, x^p_t\}_{t = 0}^{N - 1}$ are available.
\end{assumption} \begin{assumption}[{\emph{Pre-collected data}}]\label{as:statetre}
		Let $\{u^p_t\}_{t = 0}^{N - 1}$ be any sequence persistently exciting of order $L + n_x + 1$, where  $L \ge 2$.
		The state sequence  $\{x^p_{t}\}_{i = 0}^{N - 1}$ is generated by system \eqref{eq:statesys} offline using the input sequence $\{u^p_t\}_{t = 0}^{N - 1}$, from any initial condition $x^{p}_0$.
	\end{assumption}

	Different from the (data-driven) MPC, (data-driven) state-feedback control does not produce a system trajectory.  Hence, we resort to the following problem to predict the states for ensuing $L$ times at each triggered time $t_\ell$
	\begin{subequations}\label{eq:statempc}
		\begin{align}
		J^*_L(x_{t_\ell}) :=\underset{g(t_\ell), h(t_\ell)\atop
			\bar{x}_i(t_\ell)}{\max}
		&\sum_{i = 0}^{L - 1}\Vert \bar{x}_i(t_{\ell})\Vert_{\infty} \label{eq:statempc0}\\
		{\rm s.t.}\quad 
		&\left[
		\begin{matrix}
		u(t_\ell)\\
		\bar{x}(t_\ell)
		\end{matrix}
		\right] = 
		\left[
		\begin{matrix}
		H_{L}(u^p)\\
		H_{L}(x^p)
		\end{matrix}
		\right] g(t_\ell) \label{eq:statempc1}\\
		& \,\bar{x}_{0}(t_\ell) + h(t_\ell) = \zeta_{t_\ell} \label{eq:statempc2}\\
		& \, \Vert h(t_\ell)\Vert \le \bar{n} \label{eq:statempc3}
		\end{align}
	\end{subequations}
	where $\bar{x}(t_\ell) = [\bar{x}_{0}^\top(t_\ell)~ \cdots~ \bar{x}^\top_{L - 1}(t_\ell)]^{\top} \in \mathbb{R}^{n_xL}$ collects the predicted states for times from $t_\ell$ to $t_\ell + L -1$, $u(t_\ell) = [u_{t_\ell}^\top~ \cdots~ u_{t_\ell}^\top]^{\top} \in \mathbb{R}^{n_uL }$ repeats the control input $u_{t_\ell}$ for all times in $[t_\ell, t_\ell + L -1]$, and $h(t_\ell) \in \mathbb{R}^{n_x}$ is a slack variable mitigating the noise $n_t$.
	 Upon solving \eqref{eq:statempc}, the optimal state $\bar{x}^*(t_\ell)$  predicts the worst  (in terms of norm) state trajectory under noise $n_{t}$ and control inputs $u(t_\ell)$, which is subsequently used to characterize the actual future state following the same procedure as in Section \ref{sec:outputddstmpc}.
	Specifically, capitalizing on  $\bar{x}^*(t_\ell)$, the self-triggering mechanism is constructed as follows
	\begin{equation}\label{eq:selftri}
	\left\{
	\begin{aligned}
	t_{\ell + 1} &:= t_{\ell} + \min{\{L - 1, \tau_{\ell}\}},~~ t \in \mathbb{N}_{0};~~ t_0 := 0\\
	\tau_{\ell} &:= \min\{\tau\in \mathbb{N}: \varphi(\bar{x}_{\tau}(t_\ell),\, \zeta_{t_\ell},\, \rho^{\tau}\!,\,\bar{n},\, \tau) > \sigma \Vert x_{t_\ell}\Vert_{\infty}\}
	\end{aligned}
	\right.
	\end{equation}
	where $\sigma \in(0,1)$ is a threshold parameter, $\rho^{\tau_\ell} := \min\{\rho| \rho >= A^{\tau_{\ell}} \}$ for $\tau = 1, \cdots, L - 1$ is defined in Lemma \ref{lem:outputerror} with $C = I$, and $\varphi:\mathbb{R}^{n_x} \times \mathbb{R}^{n_x} \times \mathbb{R}\times \mathbb{R}\times \mathbb{N} \rightarrow \mathbb{R}_{> 0}$ is a certain function to quantify the error induced by the self-triggered sampling.
The error between the states at the current triggered time $t_\ell$ and its next one satisfies
	\begin{equation}\label{eq:xtell-xtell+1}
	\Vert \zeta_{t_{\ell}} - x_{t_{\ell+ 1}}\Vert_{\infty} \le \Vert x_{t_{\ell}} - \bar{x}^*_{\tau_{\ell}}(t_{\ell})\Vert_{\infty}+ \Vert x_{t_{\ell + 1}} - \bar{x}^*_{\tau_{\ell}}(t_{\ell})\Vert_{\infty} + \bar{n}.
	\end{equation}
	It follows from Lemma \ref{lem:outputerror} that the error between the predicted state and the actual state 
can be bounded by
	\begin{align}\label{eq:barxtell+1-xtell+1}
	\Vert x_{t_{\ell + 1}}- \bar{x}^*_{\tau_{\ell}}(t_{\ell})\Vert_{\infty} \le \rho^{\tau_\ell}(\bar{n} +\Vert h^*(t_{\ell}) \Vert_{\infty}).
	\end{align}
	Substituting \eqref{eq:barxtell+1-xtell+1} into \eqref{eq:xtell-xtell+1}, the self-triggering function $\varphi(\cdot)$ in \eqref{eq:selftri} can be given by
	\begin{align}\label{eq:phi}
	\varphi(\bar{x}(t_\ell),&~ h(t_\ell), \zeta_{t_\ell}, \rho^\tau, \bar{n}, \tau)\\
	& := \Vert x_{t_\ell} -\bar{x}_{\tau}(t_{\ell})\Vert_{\infty} + \rho^{\tau_\ell}(\bar{n} +\Vert h(t_{\ell}) \Vert_{\infty}) + \bar{n} \nonumber.
	\end{align}
	
	Based on \eqref{eq:selftri}, \eqref{eq:barxtell+1-xtell+1}, and \eqref{eq:phi}, the following stability result comes ready whose proof is similar to that of the model-based self-triggered control in \cite{Wakaiki2021selftrigger} and is thus omitted here.
	\begin{lemma}
		Let Assumptions \ref{as:ctrl}, \ref{as:statesys}, and \ref{as:statetre} hold.
		Consider the state feedback system in \eqref{eq:statesys}.
		If the self-triggered times satisfy \eqref{eq:selftri} with function $\varphi(\cdot)$ defined by \eqref{eq:phi}, and the parameters $h^*(t_\ell)$, $\bar{x}^*(t_\ell)$ obtained by solving problem \eqref{eq:statempc}, then system \eqref{eq:statesys} achieves input-to-state stability under the proposed self-triggering mechanism.
	\end{lemma} 

\begin{remark}[{\emph{Novelty}}]
	The closest references to the results in this section are \cite{Wildhagen2021data} and \cite{Wang2021data}, yet they distinguish from the  present work considerably. 
First, the work \cite{Wildhagen2021data} is devoted to estimating the maximum sampling interval for state feedback systems under different levels of process noise and controller gains.
	Therefore, they only use the pre-collected data to estimate upper bounds of $\Vert A^i \Vert$ for all $i \in \{1, \cdots,\bar{h}\}$.
	This plays a similar role as the norm maximization problem \eqref{eq:rho}. 
	However, the key idea of our self-triggering mechanism is to construct a future state trajectory using the most recently received state, and \eqref{eq:rho} offers a way for estimating one of the parameters required for constructing the trajectory.
	In addition, although the work \cite{Wang2021data} considers the self-triggered control problem by generalizing the estimation method in \cite{Wildhagen2021data}, its system model as well as associated self-triggered controller design are completely different from ours. 	
\end{remark}

	\section{Numerical Examples}
	To validate the effectiveness of our proposed data-driven self-triggering mechanism as well as predictive controllers, several numerical examples are provided in this section. 
	
	\subsection{Output feedback system}\label{sec:num:output}
	Consider a linearized version of the four-tank system in, e.g., \cite{Raff2006Nonlinear,berberich2019data}.
	The system matrices are given by 
	\begin{align*}
	&A = \left[
	\begin{matrix}
	0.927 & 0 & 0.041 & 0\\
	0 & 0.918 & 0 & 0.033\\
	0 & 0 & 0.924 & 0\\
	0 & 0 & 0 & 0.937
	\end{matrix}
	\right], 
	B = \left[
	\begin{matrix}
	0.017 & 0.001\\
	0.001 & 0.023\\
	0 & 0.061\\
	0.072 & 0
	\end{matrix}
	\right]\\
	&C = \left[
	\begin{matrix}
	1 & 0 & 0 & 0\\
	0 & 1 & 0 & 0
	\end{matrix}
	\right], \quad D = 0.
	\end{align*}
	with the observability index $\eta = 2$.
	To this end, an input-output trajectory of length $N = 800$ was obtained by means of simulating the open-loop system off-line using persistently exciting input sequence.
	Let the network-induced noise $n_t$ satisfies $\bar{n} = 0.0015$.
	Parameters of Alg. \ref{alg:ddmpc} were set as follows: the prediction horizon $L = 11$, cost matrices $Q = I$, $R = 8 \cdot 10^{-3}I$, coefficients $\lambda_g\bar{n} = 10^{-6}$, $\lambda_h/\bar{n} = 500$ the input constraint set $\mathbb{U} = [-2, 2]^2$, and the input-output equilibrium $(u^e, y^e) = ([1,1]^\top, [0.65,0.77])^\top$.
	Over a simulation horizon of $200$ time steps, Fig. \ref{fig:yt} depicts the output trajectory of system using data-driven self-triggering MPC scheme with $\sigma = 0.88$.
	Evidently, 
	 the system converges to the setpoint with $37$ measurement packets transmitted to the controller.
	In addition, since $\eta = 2$,  at most $2\times 37 = 74$ outputs are sent based on the packetized transmission protocol in Section \ref{sec:preliminaries:ncs}, which is also much lesser than $200$, and illustrates the effectiveness of the self-triggering mechanism.
	Figs. \ref{fig:compareJ} and \ref{fig:comparesigma4} compare the cost function $J_L^*(t_\ell)$ and the self-triggered times for different $\sigma$ values. 
	\begin{figure}
		\centering
		\includegraphics[width=9cm]{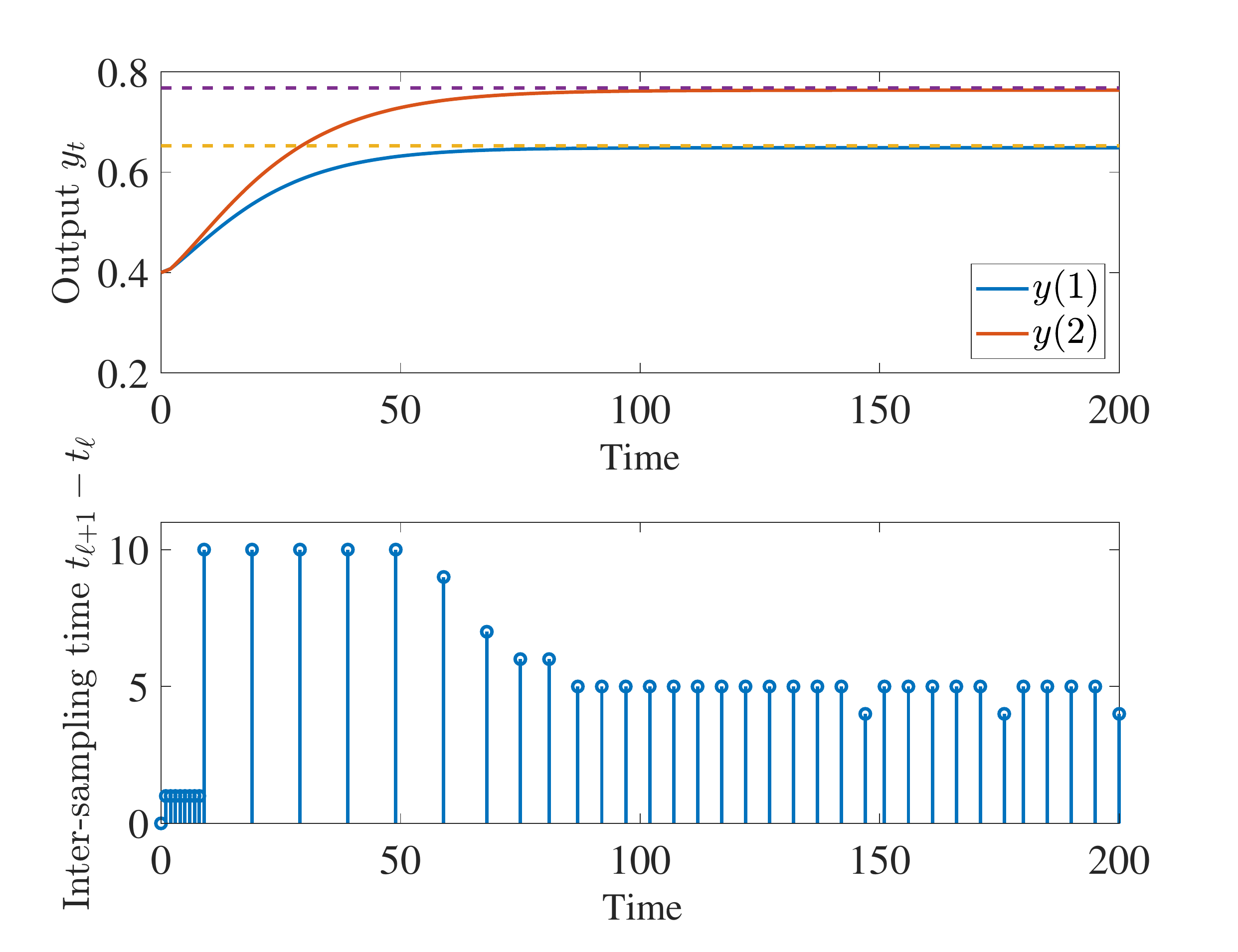}\\
		\caption{Output trajectory (top) and inter-sampling time $t_{\ell + 1} - t_{\ell}$ (bottom) with $\sigma = 0.88$.}\label{fig:yt}
		\centering
	\end{figure}
	\begin{figure}
		\centering
		\includegraphics[width=9cm]{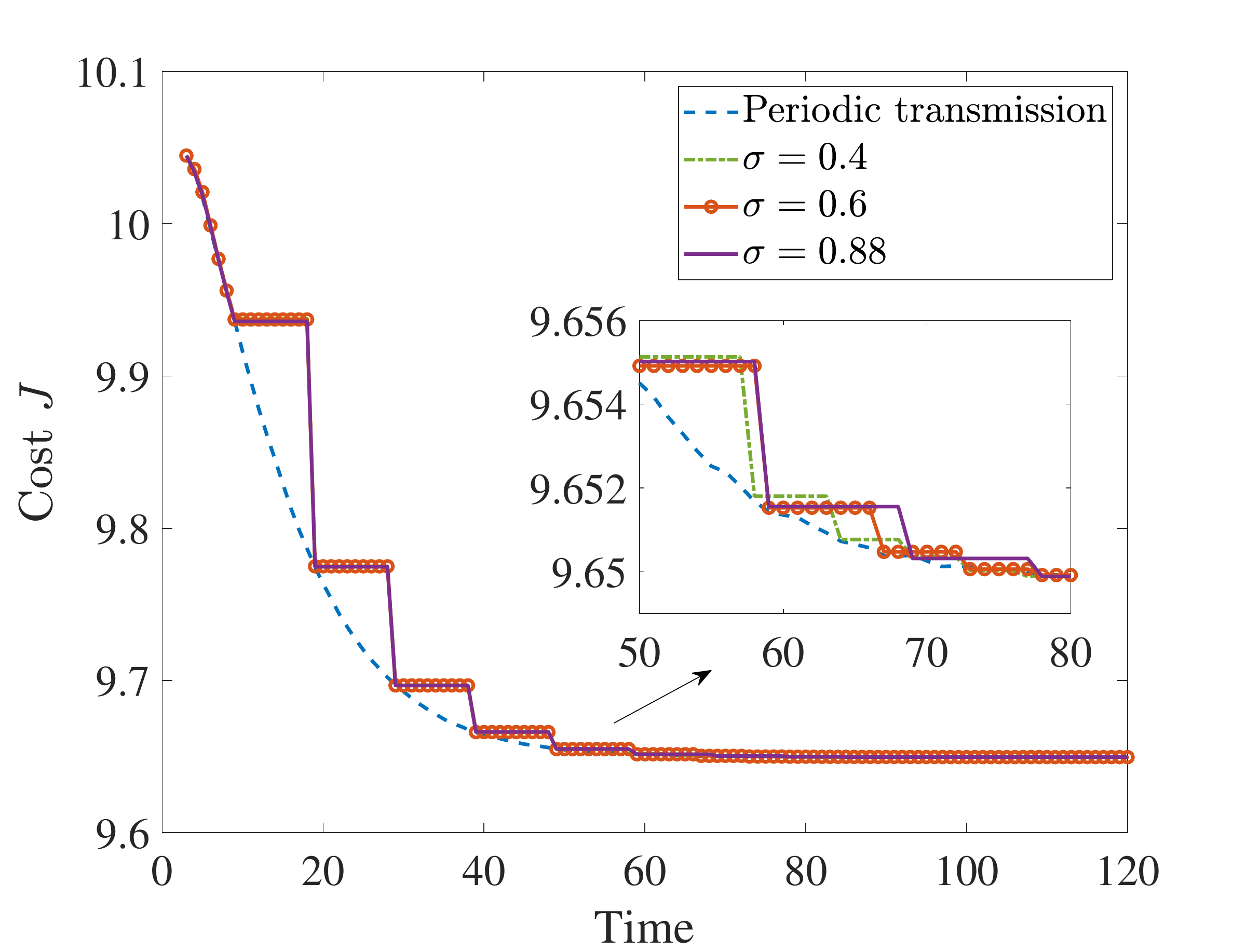}\\
		\caption{Optimal cost $J^*_L$ with different $\sigma$.}\label{fig:compareJ}
		\centering
	\end{figure}
	\begin{figure}
		\centering
		\includegraphics[width=9cm]{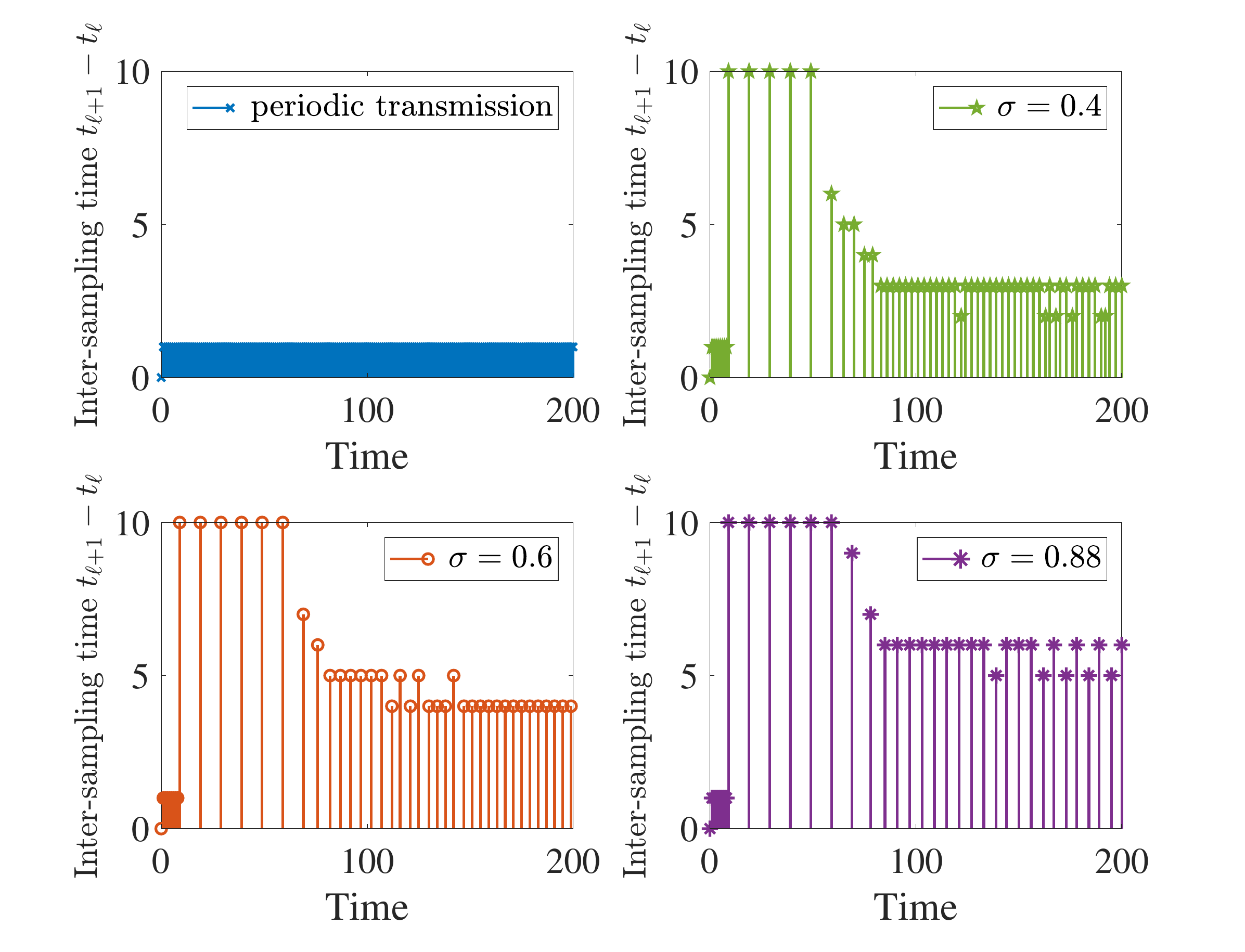}\\
		\caption{Inter-sampling time $t_{\ell + 1} - t_\ell$ with different $\sigma$.}\label{fig:comparesigma4}
		\centering
	\end{figure}	
	\subsection{State feedback systems}\label{sec:num:state}
	In this part, we consider state feedback systems and illustrate the effectiveness of the self-triggering mechanism in Section \ref{sec:stateddmpc}.
	To guarantee the system performance as done in the model-based event-triggered case e.g., \cite{peng2013anovel},  the following $\mathcal{L}_2$-like condition is incorporated and examined  
	\begin{equation}
	\kappa \sum_{i = 1}^{\tau_\ell} \Vert\bar{x}^*_i(t_\ell)\Vert \le \mu \tau_\ell\bar{n}
	\end{equation}
	where constants $\kappa>0$ and $\mu >0$ balance between achieving a better system performance and incurring less transmissions.

	\subsubsection{Example 1}
	Consider the system \cite{Wang2021data}
	\begin{equation}
	\dot{x}(t) = \left[
	\begin{matrix}
	0 & 1\\
	0 &-0.1
	\end{matrix}
	\right] x(t) + 
	\left[
	\begin{matrix}
	0\\
	0.1
	\end{matrix}
	\right] u(t),~~ t \ge 0.
	\end{equation} 
	Here, we consider a discrete-time version of this system with a sampling period of $0.1s$.
	Let $\sigma = 0.27$, $\kappa = 0.1$, $\mu = 200$, and  use the initial state $x_0 = [3~-2]^\top$.
	System performance is depicted in Fig. \ref{fig:sigma027}, where only $14$ out of $200$ samples are needed.
	This outperforms the result in the previous work \cite{Wang2021data}, which uses $15$ measurements under noise bounded by $10^{-4}$.
	\begin{figure}
		\centering
		\includegraphics[width=9cm]{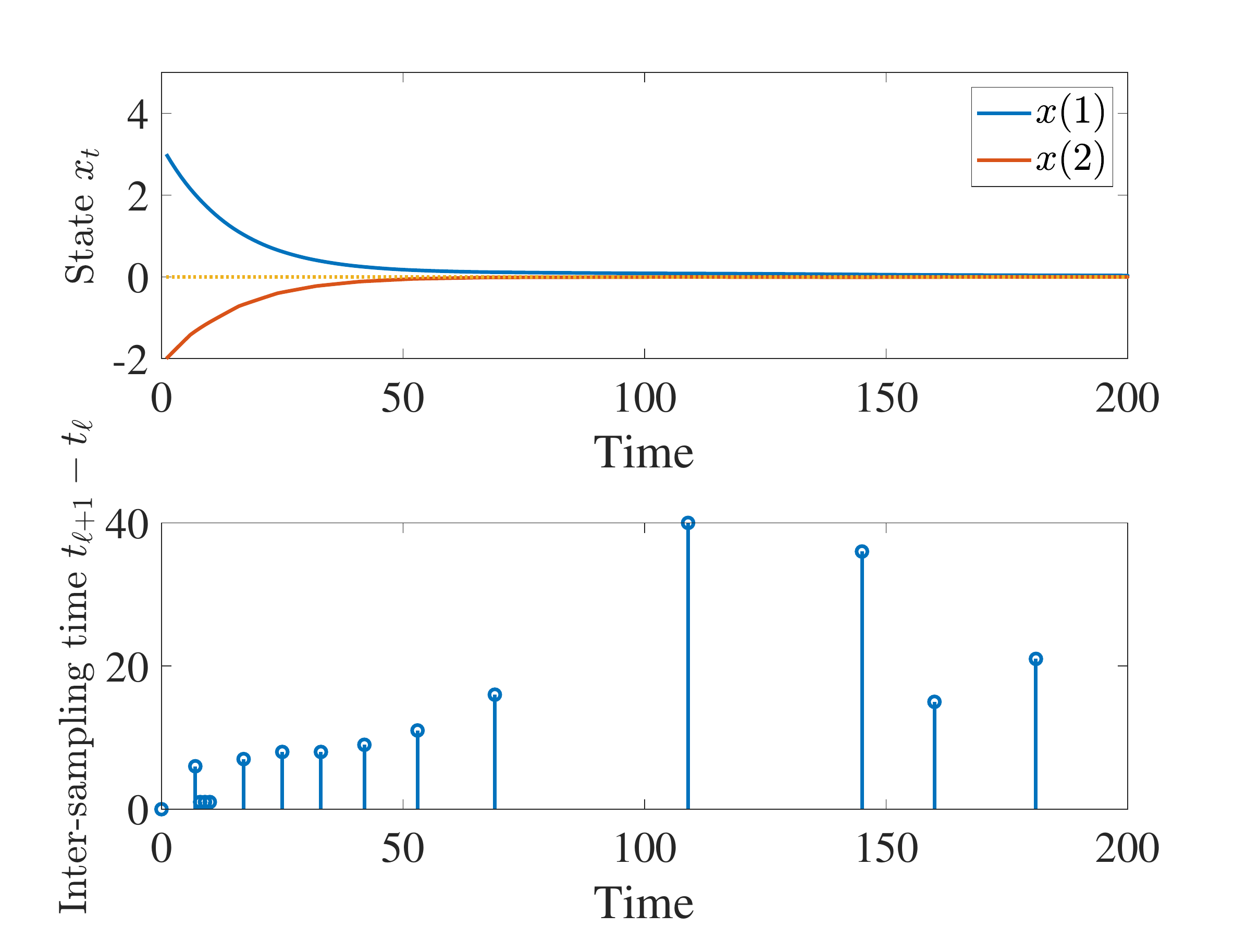}\\
		\caption{State trajectory of the order-$2$ system under \eqref{eq:statempc} with $\sigma = 0.27$.}\label{fig:sigma027}
		\centering
	\end{figure}

	\subsubsection{Example 2}
	Consider the inverted pendulum control problem in \cite{peng2013anovel}.
	The linearized system is given by
	\begin{align}
	\dot{x}(t) = \left[
	\begin{matrix}
	0 & 1 & 0 & 0\\
	0 &0 &\frac{m_1 g}{-m_2} & 0\\
	0 & 0 & 0 & 1\\
	0 & 0 & \frac{g}{\ell} & 0
	\end{matrix}
	\right] x(t) + \left[
	\begin{matrix}
	0\\
	\frac{1}{m_2}\\
	0\\
	\frac{-1}{m_2 \ell}
	\end{matrix}
	\right] u(t)
	\end{align}
	where $m_1 = 1$, $m_2 = 10$, $\ell = 3$, and $g = 10$.
	Let $\sigma = 0.27$, $\kappa = 0.1$, $\mu = 200$, and the initial state $x_0 = [0.98~0~0.2~0]^\top$.
	It can be seen from Fig. \ref{fig:sigma031_noise015_4order} that the system converges to zero with $62$ samples over $[0,200]$, which is less than that of \cite{Wang2021data}.
	\begin{figure}
		\centering
		\includegraphics[width=9cm]{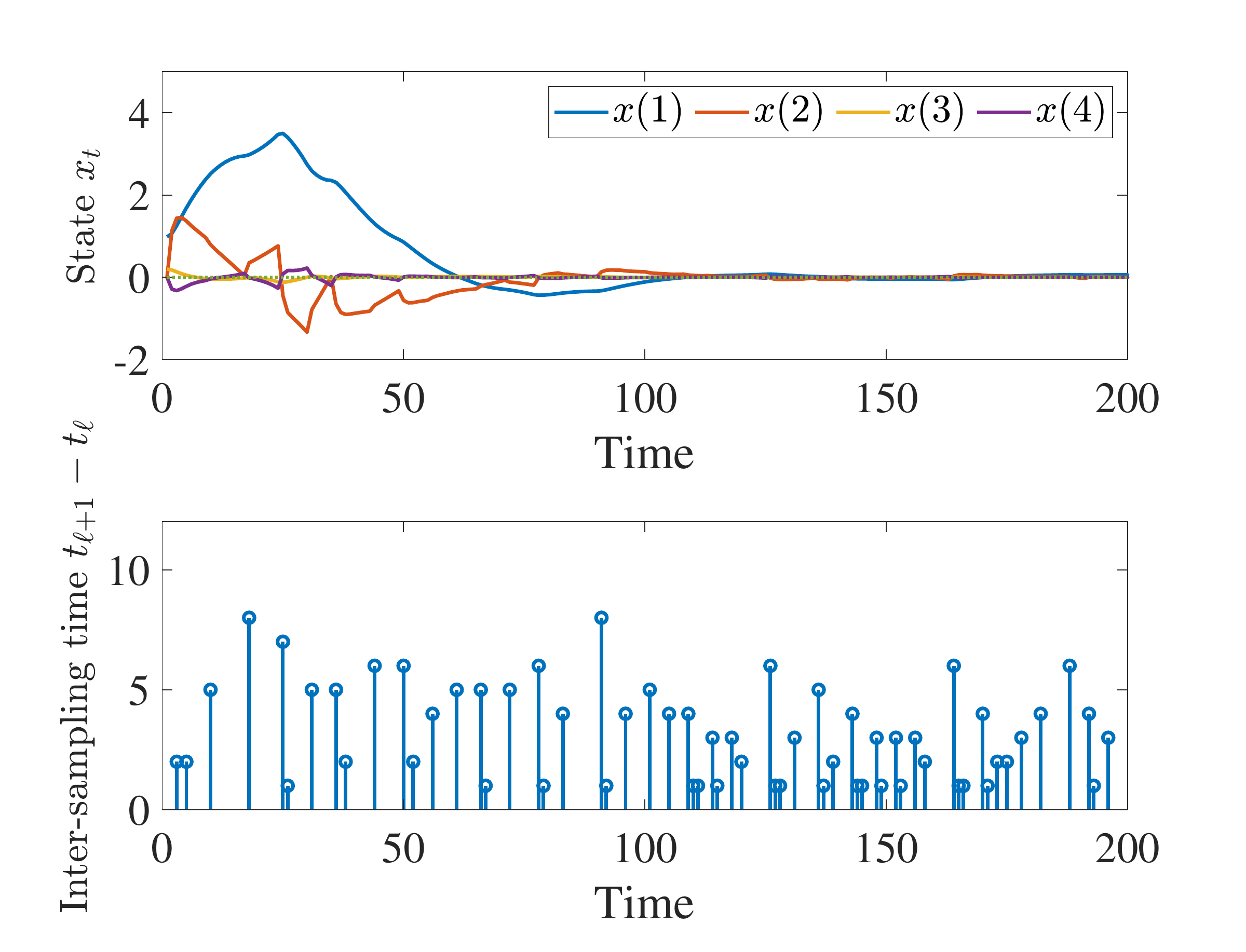}\\
		\caption{State trajectory of the order-$4$ system under \eqref{eq:statempc} with $\sigma = 0.3$.}\label{fig:sigma031_noise015_4order}
		\centering
	\end{figure}
	
	\section{Conclusions}\label{sec:conclusion}
	
	This paper puts forth a data-driven self-triggering control framework for unknown linear systems through trajectory prediction.
		For output feedback systems, a data-driven MPC scheme is developed, which generates a sequence of control inputs  once an event (transmission) is triggered. The predicted output trajectory from the MPC is further used to design a self-triggering law that is purely data-based.
	Both feasibility as well as practical exponential stability were established for the resulting data-driven self-triggered MPC. 
	Moreover, when a state feedback control law is considered, 
a norm maximization problem was developed to predict future system states, thus enabling the data-driven self-triggered control.
	Finally, numerical examples are
	presented to validate the practical merits of the proposed data-driven control methods and theory.

	\bibliographystyle{IEEEtran}
	
	\bibliography{bible3}

\end{document}